\documentclass[a4paper, 10pt]{article}
\usepackage[utf8]{inputenc}
\usepackage{jheparxiv}
\usepackage{amsfonts,amsmath,amssymb,amsthm,mathrsfs,mathtools,bm}
\usepackage[final]{showkeys}
\usepackage{comment}
\usepackage[normalem]{ulem}
\usepackage{csquotes}
\usepackage{subcaption}
\usepackage{wrapfig}
\usepackage{bigints}
\usepackage[shortlabels]{enumitem}
\RequirePackage[dvipsnames]{xcolor}
\definecolor{arXiv}{named}{OliveGreen}
\definecolor{ColorCite}{named}{BrickRed}
\definecolor{ColorLink}{named}{NavyBlue}
\definecolor{ColorURL}{named}{RoyalBlue}
\usepackage{tikz,tikz-cd} 
\usetikzlibrary{decorations.markings,decorations.pathmorphing}
\usetikzlibrary{angles,quotes} 
\usetikzlibrary{arrows.meta} 
\newtheorem{definition}{Definition}[section]

\newtheorem{proposition}{Proposition}[section]
\theoremstyle{remark}
\newtheorem{remark}{Remark}[section]

 \newcommand{\badat}{\begin{alignedat}}
 \newcommand{\eadat}{\end{alignedat}}

\newcommand\ynote[1]{\textcolor{green!60!black}{\bf [Y:\,#1]}}
\newcommand{\ymod}[1]{\textcolor{green!60!black}{#1}}

 \newcommand{\e}{{\color{red} \varepsilon}}
 \newcommand{\correc}[1]{#1}

\newcommand{\Ti}{\textnormal{Ti}}
\newcommand{\Spi}{\textnormal{Spi}}

\newcommand{\Scri}{\mathscr{I}}
\newcommand{\pM}{\hat{M}}
\newcommand{\pg}{\hat{g}}
\newcommand{\y}{\bm{y}} 
\newcommand{\R}{\mathbb{R}}
\newcommand{\SPI}{\textnormal{SPI}}
\newcommand{\BMS}{\textnormal{BMS}}

\newcommand{\scal}{\cdot} 
\newcommand{\bigO}[1]{\textrm{O(}#1\textrm{)}}
\newcommand{\smallo}[1]{\textrm{o(}#1\textrm{)}}
\newcommand{\pa}{\partial} 
\DeclareMathOperator{\sgn}{sgn}

\title{Ti and Spi, Carrollian extended boundaries at timelike and spatial infinity}

\author{Jack Borthwick$^{\dagger}$, Maël Chantreau$^{\ddagger,\star}$, Yannick Herfray$^{\ddagger}$.}

\affiliation{$^{\dagger}$ Corresponding Author, Department of Mathematics and Statistics, McGill University, Burnside Hall
805 Sherbrooke Street West, Montreal, Canada}

\affiliation{$^{\ddagger}$Institut Denis Poisson UMR 7013, Université de Tours, 
Parc de Grandmont, 37200 Tours, France}

\affiliation{$^{\star}$ ENS de Lyon, 46 allée d’Italie, 69007 Lyon, France}

\emailAdd{mael.chantreau@univ-tours.fr, jack.borthwick@mcgill.ca, yannick.herfray@univ-tours.fr}

 \abstract{
 The goal of this paper is to provide a definition for a notion of \emph{extended boundary} at time and space-like infinity which, following Figueroa-O'Farril--Have--Prohazka--Salzer, we refer to as $\Ti$ and $\Spi$. This definition applies to asymptotically flat spacetime in the sense of Ashtekar--Romano and we wish to demonstrate, by example, its pertinence in a number of situations. The definition is invariant, is constructed solely from the asymptotic data of the metric and is such that automorphisms of the extended boundaries are canonically identified with asymptotic symmetries. Furthermore, scattering data for massive fields are  realised as functions on $\Ti$ and a geometric identification of cuts of $\Ti$ with points of Minkowski then produces an integral formula of Kirchhoff type. Finally, $\Ti$ and $\Spi$ are both naturally equipped with (strong) Carrollian geometries which, under mild assumptions, enable to reduce the symmetry group down to the BMS group, or to Poincaré in the flat case. In particular, Strominger's matching conditions are naturally realised by restricting to Carrollian geometries compatible with a discrete symmetry of Spi.} 

\begin{document}
\maketitle
\newpage

\section{Introduction}

In this article we wish to promote a definition\footnote{See Definitions \ref{Definition: Spi} and \ref{Definition: Ti}, section \ref{Subsection: definition Ti/Spi}.} for \emph{extended} boundaries at timelike and spatial infinity, which, following \cite{Figueroa-OFarrill:2021sxz}, we refer to as $\Ti$ and $\Spi$\footnote{These are meant to rhyme with ``scri''.}. It applies to any asymptotically flat spacetime in the sense of Ashtekar--Romano \cite{Ashtekar:1991vb}. Our goal is to demonstrate, by example, the pertinence of this extended boundary in a number of situations.

The origin of the construction lies in the seminal work by Ashtekar and Hansen~\cite{ashtekar_unified_1978,Ashtekar:1991vb}. There, spatial infinity $\mathcal{I}^0\simeq dS^3$ was constructed as a blow up of $\iota^0$, the singular point in the conformal boundary of a four dimensional asymptotically flat spacetime where future $\mathscr{I}^+$ and past $\mathscr{I}^-$ meet \cite{penrose_asymptotic_1963,Wald:1984rg}, see figure \ref{fig:CMinkowski}. The authors then go on to consider a natural line bundle $\mathbb{R} \times \mathcal{I}^0 \to \mathcal{I}^0$ over spatial infinity. As was remarked in \cite{gibbons_ashtekar-hansen_2019}, the resulting $4$D manifold is -- weakly pseudo -- Carrollian (in the sense of \cite{duval_carroll_2014}). However, this construction was really tied up to assumptions about how future and past null infinity are joined together along a point and does not immediately extend to spacetimes which are asymptotically flat in the sense of Ashtekar--Romano \cite{beig_einsteins_1982,Ashtekar:1991vb}.

\begin{figure}[h!]
\begin{subfigure}{.5\textwidth}
\centering
\includegraphics[scale=.7]{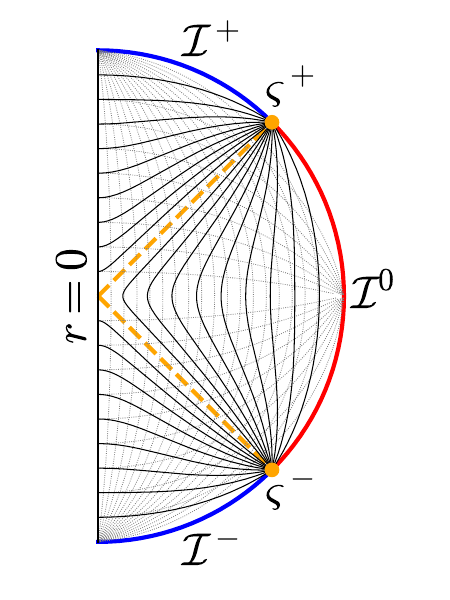}
\subcaption{\label{fig:CProjMinkowski}Projective boundaries}
\end{subfigure}%
\begin{subfigure}{.5\textwidth}
\centering
\includegraphics[scale=0.47]{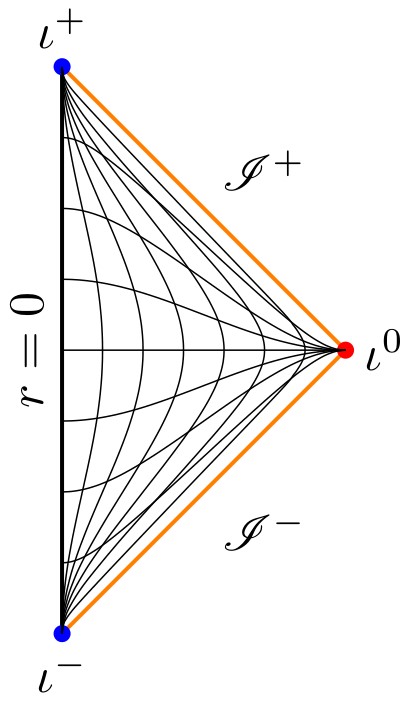}
\subcaption{Conformal boundaries\label{fig:CConfMinkowski}}
\end{subfigure}
\caption{The projective and conformal compactifications of Minkowski space $\mathbb{M}^{n+1}$.}\label{fig:CMinkowski}

\hfill\begin{minipage}{\dimexpr\textwidth-1cm}\small
 On the left (a), the projective compactification: $\mathcal{I}^{\pm} \simeq H^n$ denotes future/past projective infinity, isometric to hyperbolic space. $\mathcal{I}^0 \simeq dS^n$ denotes projective spatial infinity, isometric to de Sitter space. Projective null infinities $\varsigma^{\pm}\simeq S^{n-1}$ are (Riemannian) conformal spheres. Altogether the boundary is topologically $S^n$. The lightcone of the origin is represented by the dotted orange line, it divides the compactified spacetime into different regions that are foliated with the level hypersurfaces of $\rho=\frac{1}{\sqrt{\lvert{-t^2+x^2}\rvert}}$. The faint black horizontal (resp. vertical) curves are the $t=\textrm{cst}$ (resp. $r=\textrm{cst}$) slices.\\
On the right (b), the conformal compactification: the inside is foliated horizontally (resp. vertically) by the $t=\textrm{cst}$ (resp. $r=\textrm{cst}$) hypersurfaces.  $\mathscr{I}^{\pm} \simeq S^{n-1} \times \mathbb{R}$ denotes future/past conformal null infinity while conformal future/past $i^{\pm}$ and spatial $i^0$ infinities are reduced to points. The total boundary is topologically $S^1 \times S^{n-1}$.
\end{minipage}~\hspace{0.5cm}
\end{figure}

In most of the recent literature on spatial infinity, and following the seminal work \cite{beig_einsteins_1982,Ashtekar:1991vb,Friedrich:1998-83}, spatial\footnote{We will use the terms \enquote{space-like} and \enquote{spatial} interchangeably.} infinity $\mathcal{I}^0$ is indeed rather thought of as a co-dimension one boundary added at spatial infinity \cite{Friedrich:1999wk,Friedrich:2002ru,Kroon:2003ix,Kroon:2004pu,Kroon:2004me,Compere:2011ve,Hintz:2017xxu,troessaert_bms4_2018,Henneaux:2018mgn,AliMohamed:2021nuc,Capone:2022gme,Fuentealba:2022xsz,Mohamed:2023ttb,Mohamed:2023jwv,Hintz:2023kjj,Fuentealba:2023rvf,Fuentealba:2023hzq,Fuentealba:2023syb,Gasperin:2024bfc,Fiorucci:2024ndw,Fuentealba:2024lll} (this, however, does not prevent Ashtekar-Hansen's initial idea of a one-point completion from being useful \cite{Chrusciel:1989ye,Herberthson:1992gcz,Prabhu:2019fsp,Prabhu:2021cgk,Ashtekar:2023zul,Ashtekar:2024aa}). Similarly, past (future) timelike infinity have been investigated in \cite{Gen:1997az,Cutler:1989,Porrill:1982,Campiglia:2015kxa,Campiglia:2015qka,Campiglia:2015lxa,Chakraborty:2021sbc,Compere:2023qoa}.  This is similar in spirit to Penrose's conformal compactification, allowing one to treat asymptotics of the fields in terms of local differential geometry close to $\mathcal{I}^0$, but, as we will review later on, rather genuinely belongs to the realm of projective compactification as introduced by Čap and Gover in \cite{Cap:2014aa,Cap:2014ab}; see Figure \ref{fig:CMinkowski} for comparison and see \cite{eastwood_bgg_2011,cap_projective_2012,cap_holonomy_2014,cap_projective_2016-1,Gover:2018vui,Borthwick_Proca_Projective,Flood:2018aa,RSTA20230042,Borthwick:2024wfn} for related works. 

Nevertheless, the (projective) boundaries $\mathcal{I}^-, \mathcal{I}^0,\mathcal{I}^+$ at past/spatial/future infinity are too rigid. For example, one cannot realise the asymptotic symmetry group, i.e. the $\SPI$ group from \cite{Ashtekar:1991vb}, as some intrinsic symmetry group of $\mathcal{I}^-, \mathcal{I}^0,\mathcal{I}^+$ alone. A closely related point is that the scattering data for massive fields are not, strictly speaking, functions on $\mathcal{I}^\pm$. The extended boundaries $\Ti^{\pm} \simeq \mathbb{R} \times \mathcal{I}^{\pm}$ and $\Spi \simeq \mathbb{R} \times \mathcal{I}^{0}$ will solve these problems.

\paragraph{Extended boundary at spatial and timelike infinity.}\mbox{}

In introducing these extended boundaries, we are heavily influenced by the work~\cite{Figueroa-OFarrill:2021sxz} of Figueroa O'Farrill, Have, Prohazka and Salzer which, in the flat case of Minkowski, described natural line bundles over spatial and time infinity. They introduced the names Ti and Spi for the (pseudo) Carrollian homogeneous spaces associated to the Poincaré group $\textrm{ISO}(3,1) = \textrm{SO}(3,1) \ltimes \mathbb{R}^{4}$:
\begin{align}\label{Intro: flat models}
\begin{alignedat}{2}
        \Ti^{(3)} = \frac{\textrm{SO}(3,1) \ltimes \mathbb{R}^{4}}{ \textrm{SO}(3) \ltimes \mathbb{R}^{3}} &\to  H^3 = \frac{\textrm{SO}(3,1) }{ \textrm{SO}(3)},\\[0.5em]
        \Spi^{(3)} = \frac{\textrm{SO}(3,1) \ltimes \mathbb{R}^{4}}{ \textrm{SO}(2,1) \ltimes \mathbb{R}^{3}} &\to  dS^3 = \frac{\textrm{SO}(3,1) }{ \textrm{SO}(2,1)}.
\end{alignedat}
\end{align}
These homogeneous spaces naturally fibre over the (projective) boundaries $I^{\pm} \simeq H^3$, $I^{0} \simeq dS^3$ and it was conjectured in \cite{Figueroa-OFarrill:2021sxz} that they are related to the construction of Ashtekar--Hansen \cite{ashtekar_unified_1978}. This situation was generalised to a curved setting in~\cite{RSTA20230042,Borthwick:2024wfn}, under the implicit assumption that the mass aspect was vanishing. In the present article, we will relax this assumption and therefore complete the program of describing $\Ti$ and $\Spi$ as canonical extensions of time/spatial infinity. 

We shall prove that, in the flat case, our definition for $\Ti$ and $\Spi$ implies that they are naturally equipped with a transitive action of the Poincaré group, induced by the action of the bulk, and are therefore canonically identified with the homogeneous space \eqref{Intro: flat models}. Furthermore, we shall see that this definition can be related to that of Ashtekar--Hansen \cite{ashtekar_unified_1978} in the curved case and, consequently, that it bridges between the two approaches.\\

\paragraph{Integral formula for massive fields.}\mbox{}
 In Minkowski space-time, the intersection of the null cone emanating from a given spacetime point $X^{\mu}$ with $\Scri^{\pm}$ defines a cut $S^2 \hookrightarrow \Scri \simeq \mathbb{R} \times S^2$. The value $\Phi(X)$ of a massless field at $X^{\mu}$ is then given by some integral of the scattering data \emph{on this section} $S^2 \subset \Scri$; this is the Kirchhoff-d'Adhémar formula \cite{Kirchhoff:1883aa,Penrose_1980,penrose_spinors_1984} which encodes causality. In Section \ref{Section: An integral formula on Ti} we will generalise these results to past/future time infinity and massive fields. We will show that the set of timelike geodesics passing through $X^\mu$ asymptotically defines a cut $H^3 \subset \Ti$ of $\Ti \simeq \mathbb{R} \times H^3$, that massive fields induce scattering data at Ti and that the value $\Phi(X)$ of the massive field at $X^{\mu}$ can be recovered by an integral formula over the corresponding section $H^3 \subset \Ti$, see Proposition \ref{Proposition: Reconstruction formula}. 

Once more, we are inspired by \cite{Figueroa-OFarrill:2021sxz} and \cite{Have:2024dff}: in the former, cuts of $\Ti$ and $\Spi$ were introduced, however, without relating it to the geometry of spacetime geodesics and, in the latter, it was shown that massive unitary representations of the Poincaré group can be naturally realised on the homogeneous spaces \eqref{Intro: flat models} and that the corresponding Carrollian fields relate to Fourier modes of massive fields. However, in this approach, and since $\Ti$ was again defined as an abstract homogeneous space with \emph{a priori} no relationship to Minkowski space, the dictionary between Carrollian fields and spacetime massive fields can only be tautological. This is in contrast with the present work, in which Ti is defined in terms of the asymptotic geometry. In particular, since our construction makes sense on general grounds, the results are expected to extend to the boundary value of a massive field propagating on an asymptotically flat spacetime.\\

\paragraph{Asymptotic symmetries and Carrollian geometry.}\mbox{}
Finally, Definitions \ref{Definition: Spi} and \ref{Definition: Ti} for $\Spi/\Ti$ will enable us to generalise Ashtekar's description of radiative data at null infinity \cite{geroch_asymptotic_1977, ashtekar_radiative_1981,ashtekar_geometry_2015} in terms of (an equivalence class of) Carrollian connections at $\Scri$; see \cite{Herfray:2021qmp} for the precise relation to Carroll geometries and \cite{herfray_asymptotic_2020,herfray_tractor_2022} for a conformally invariant realisation in terms of Cartan geometries. The presence of gravitational radiation at  $\Scri$ is characterised by curvature, while flatness is equivalent to the possibility of picking up a Poincaré group $\textrm{ISO}(3,1) \subset \BMS_4$ inside the BMS group.

Similarly, we will show that $\Spi/\Ti$ is canonically equipped with a Carrollian connection. When flat, it allows us to single out a Poincaré group $\textrm{ISO}(3,1) \subset \SPI_4$, this time inside the $\SPI$ group. In the presence of curvature, this connection encodes late-time asymptotics of the gravitational fields and allows us to define a BMS group $\BMS_4 \subset \SPI_4$ inside the $\SPI$ group. At $\Ti$, this is straightforward and, along the way, yields a representation of the BMS group on massive fields. In fact, it realises the Longhi-Materassi representation \cite{Longhi:1997zt} geometrically. At $\Spi$, the parity conditions studied in \cite{troessaert_bms4_2018,Henneaux:2018cst,Henneaux:2018gfi,Henneaux:2018mgn,Henneaux:2019yax,Fuentealba:2022xsz,AliMohamed:2021nuc,Capone:2022gme,Mohamed:2023jwv,Compere:2023qoa,Mohamed:2023ttb}, which are known to imply Strominger's matching conditions \cite{strominger_bms_2014}, are very naturally realised by requiring that a certain discrete symmetry of the model is preserved in the curved case.\\

The article is organised as follows. In Section \ref{Section: Ashtekar-Romano spacetimes}, we review the definition of asymptotically flat spacetimes from \cite{Ashtekar:1991vb}. In Section \ref{Section: Ti and Spi}, we introduce the definition of Ti and $\Spi$ and demonstrate that it neatly accommodates both asymptotic symmetries and scattering data of massive fields. In Section \ref{Section: An integral formula on Ti}, we establish, at Ti, the integral formula that generates the massive field from its scattering data. Finally, in Section \ref{Section: Carrollian geometry of Spi}, we show that $\Spi/\Ti$ are examples of Carrollian geometries and use the corresponding Carrollian connection to geometrically realise the restriction of the $\SPI$ group to BMS and Poincaré.\\

\subsection*{Acknowledgements}
Research supported by NSERC grant RGPIN 105490-2018.

\newpage
\section{Ashtekar-Romano spacetimes}\label{Section: Ashtekar-Romano spacetimes}

Let us begin by reviewing, from \cite{beig_einsteins_1982,Ashtekar:1991vb}, the features of asymptotically flat spacetimes at spatial/time infinity that we will need in the rest of the article. We will primarily follow the notation and conventions of \cite{Compere:2023qoa}, which very nicely gathers together and improves upon results that are scattered in the literature. Nevertheless, our main definition for asymptotically flat spacetimes at spatial/time infinity will be different (but equivalent). We believe that this new definition, heavily influenced by the work of Čap and Gover \cite{Cap:2014aa,Cap:2014ab} on projective compactification of Einstein metrics,  is not only more compelling geometrically, but also more economical.

\subsection{Definition}

\correc{In what follows, the generic symbol $\mathcal{I}$ will stand for either $\mathcal{I}^{-}$, $\mathcal{I}^{0}$ or $\mathcal{I}^{+}$; }we will always assume that $M$ is of dimension $d=n+1$, consequently, these boundaries are of dimension $n$.\

\begin{definition}\label{Definition:AsymptoticFlatness}
    A spacetime ($\pM$, $\pg_{ab}$) will be said to be asymptotically flat at, respectively, past time-like, space-like or future time-like infinity, if it can be embedded onto the interior of a manifold $M$ with boundary respectively denoted by $\mathcal{I}^{-}$, $\mathcal{I}^{0}$ or $\mathcal{I}^{+}$, and there is a boundary defining function $\rho$
    \begin{align}
        \rho\vert_{\mathcal{I}} &= 0, &  \nabla_a\rho\vert_{\mathcal{I}} & \neq 0,
    \end{align}
    such that $\pg$ is Einstein up to order $2$, i.e. \[ \hat{G}_{ab} = O(\rho^2),\] and the rescaled inverse metric:
    \begin{equation}\label{eq:rescaled_inverse_metric}
        g^{ab} := \rho^{-2} \pg^{ab}
    \end{equation}
    satisfies the following axioms.
    \begin{enumerate}
        \item $g^{ab}$ extends\footnote{See equation \eqref{Ashtekar-Romano: metric Coordinates expansion} for the typical level of differentiability that we will implicitly assume here.} to $M$ and is invertible on $M \setminus \mathcal{I}$. \label{AFSTCondition1}
        \item $N^a := \rho^{-2} g^{ab} \nabla_b \rho$ and $\nu := N^a \nabla_a \rho$ have finite non-vanishing limits on $\mathcal{I}$. \label{AFSTCondition2}
        \item The restriction $h^{ab} := g^{ab}\vert_{\mathcal{I}}$, which defines a metric on $\mathcal{I}^0/\mathcal{I}^\pm$ by the previous assumptions, is invertible and has Euclidean signature $(+ + \dots +)$ for past/future time-like infinity $\mathcal{I}^{\pm}$ (resp. Lorentzian signature $(- + \dots +)$ for space-like infinity $\mathcal{I}^0$). \label{AFSTCondition3}
    \end{enumerate}
\end{definition}

One can check that, up to regularity considerations, these assumptions are equivalent to those required by Ashtekar--Romano to have an asymptote at (spatial, in their case) infinity~\cite{Ashtekar:1991vb}. Note that the second assumption implies that, introducing a local chart $(\rho, y^{\alpha})$, where $y^{\alpha}$ are coordinates on $\mathcal{I}^{0}$/ $\mathcal{I}^{\pm}$ extended to $M$ by the requirement $N^a \nabla_a y^{\alpha}=0$, one can write
\begin{equation}
g^{ab} =  \nu \rho^2  \partial_\rho \partial_\rho + \gamma^{\alpha \beta} \partial_\alpha \partial_\beta, \label{Eq::InverseMetric}
\end{equation}
and this implies that we have an asymptotic expansion for the metric of Beig-Schmidt type \cite{beig_einsteins_1982}. Following the conventions and assumptions of the recent works \cite{Compere:2011ve,Compere:2023qoa}, we will assume that $g_{ab}$ has the following behaviour:
\begin{align}\label{Ashtekar-Romano: metric Coordinates expansion}
\badat{2}
    g_{ab} =& \rho^{-2}\; \nu_0^{-1}\,\Big(  1 +2 \rho \sigma +\rho^2 \sigma^2 + O(\rho^3 \ln \rho)  \Big) d\rho^2 \\
    &+ \Big( h_{\alpha\beta} + \rho \left( k_{\alpha\beta} -2 \sigma h_{\alpha \beta}\right) +  O(\rho^2 \ln \rho) \Big) d y^{\alpha} d y^{\beta};
\eadat
\end{align}
the potential logarithmic terms hidden in the remainders will play no role in this article. 
We deduce that:
\begin{align}\label{Ashtekar-Romano: Coordinates expansion}
    \nu &= \nu_0\left( 1  - 2 \rho \sigma \right) + o(\rho^{2}),&
    \gamma^{\alpha\beta} &= h^{\alpha\beta} - \rho \left( k^{\alpha\beta} -2 \sigma h^{\alpha \beta}\right) + O(\rho^2\ln\rho).
\end{align}
In the above, $h_{\alpha \beta}$ is the inverse of $h^{\alpha \beta}$ and this metric along the boundary is used to lower and raise indices $\alpha,\beta, \gamma, \dots$ from the beginning of the Greek alphabet. Also note that:
\begin{align}
    \sgn(\nu_0) &= +1 \qquad \text{on} \quad \mathcal{I}^{0}, & \sgn(\nu_0) &= -1 \qquad \text{on} \quad \mathcal{I}^{\pm},
\end{align}
where again $\mathcal{I}^{-}$/$\mathcal{I}^{+}$, $\mathcal{I}^{0}$ respectively stand for past/future time-like and spatial infinity.

\paragraph{Einstein's equations}\mbox{}

Here we give an overview of some of the consequences of Einstein's equations on $\pg$.
Since Definition \ref{Definition:AsymptoticFlatness} is equivalent to that of Ashtekar--Romano, the results summarised below are equivalent to those in the literature \cite{Compere:2011ve,Compere:2023qoa}. We will, however, work in arbitrary dimension $d = n+1$, $n > 1$. 
$\hat{G}_{a b} = \bigO{\rho^2}$ implies that the first two orders of $\hat{R}_{ab}$ vanish. To leading order in $\rho$, one obtains:
\begin{equation}
\begin{aligned}\label{Ashtekar Romano: Einstein equations on h}
    &D_\alpha \nu_0 = 0 \\
    &R^{(h)}{}_{\alpha \beta} = \nu_0 (n-1) h_{\alpha \beta},
\end{aligned}\end{equation}
where $D_\alpha$ and $R^{(h)}{}_{\alpha \beta}$ are, respectively, the covariant derivative and the Ricci tensor associated to $h_{\alpha \beta}$. The first equation in~\eqref{Ashtekar Romano: Einstein equations on h}, combined with the ambiguity in the boundary defining function that will be emphasised in the next section, means that $|\nu_0|$ can be prescribed arbitrarily. In the literature, it is common practice to assign $|\nu_0|=1$; we will divert slightly from this by  allowing it to take an arbitrary positive value. 
The second equation is the vacuum Einstein equation on $\mathcal{I}^{0}$/ $\mathcal{I}^{\pm}$. In dimensions $n \leq 3$, the Weyl tensor vanishes identically and it imposes that $\mathcal{I}^{0}$/ $\mathcal{I}^{\pm}$ have constant sectional curvature $\kappa = \nu_0$. In higher dimensions, this is no longer necessary but we will nevertheless assume that past/future timelike infinity $\mathcal{I}^{\pm}$ is (locally) hyperbolic space $H^n$, and spacelike infinity is (locally) De Sitter space $dS_n$; this will simplify the presentation but a number of the results in this article will not rely on this assumption. 

At next order in $\rho$, Einstein's equations on $\pg$ yield:
\begin{equation}
    (D^2 + n \nu_0) \sigma =0, \label{Ashtekar Romano: Einstein equations on sigma}
\end{equation}
and
\begin{align*}
 &h^{\gamma\beta}D_{[\beta} k_{\alpha]\gamma} =0, \\
    &\frac{1}{2}(D^2k_{\alpha\beta}-n\nu_0k_{\alpha\beta}-D_{\alpha}D_{\beta} k + \nu_0h_{\alpha\beta}k) -{W^{(h)}}_{\phantom{\delta}(\alpha\beta)\gamma}^{\delta}k^\gamma_{\,\delta}=(n-3)\left(D_{\alpha}D_{\beta}\sigma + \nu_0h_{\alpha\beta} \sigma\right). %
\end{align*}
The first and third equations are respectively evolution equations on $\sigma$ and $k_{\alpha \beta}$, while the second is a constraint on $k_{\alpha \beta}$. It is interesting to note that it is only possible to separate the equations when $n=3$; this is achieved through the definition of $k_{\alpha \beta}$. In other dimensions, $\sigma$ will source the evolution on $k_{\alpha \beta}$. We briefly observe that the last two equations can be rewritten in the more compact form:
\begin{gather}
   \begin{cases}h^{\gamma\beta}D_{[\beta} k_{\alpha]\gamma} =0, \\[0.5em]
h^{\gamma\delta}\left(D_{\gamma} D_{[\delta}k_{\alpha]\beta} + \frac{1}{2}k_{\gamma\epsilon}{W^{(h)}}_{\phantom{\epsilon}\beta\delta\alpha}^{\epsilon}\right) =(n-3)(D_\alpha D_\beta \sigma + \nu_0h_{\alpha\beta}\sigma), \label{Ashtekar Romano: Einstein equations on k}
    \end{cases}
\end{gather}
 These equations have some important invariance properties in terms of projective geometry, that we will discuss in a subsequent work.

\paragraph{Ambiguity in the boundary defining function}\mbox{}

\noindent Definition~\ref{Definition:AsymptoticFlatness} of asymptotical flatness at time/space-like infinity guarantees the existence of a boundary defining function $\rho$, such that conditions \ref{AFSTCondition1}-\ref{AFSTCondition3} are satisfied. However, these conditions do not uniquely determine $\rho$. For instance, if Condition~\ref{AFSTCondition1} holds for $\rho$ it will also hold for \emph{any} other boundary definition function 
\begin{equation*}
    \rho' = \lambda \rho = \big(\lambda_0 + O(\rho)\big) \rho,
\end{equation*}
where $\lambda_0$ is a nowhere vanishing function on $\mathcal{I}^{\pm}/\mathcal{I}^{0}$. Condition~\ref{AFSTCondition2}, on the other hand, requires that $\lambda_0$ must be a \emph{constant} function on $\mathcal{I}^{\pm}/\mathcal{I}^{0}$. This constant can be fixed to one, if we require that the scalar curvature of the boundary metric $h_{\alpha \beta}$ be fixed to some constant value $n(n-1)\nu_0$ or, equivalently, that $\nu_0$ is kept fixed under change of boundary defining function.
Overall, with this minor restriction, we see that the boundary defining function $\rho$ is defined up to:
\begin{equation}
   \rho \quad \mapsto \quad \rho(1 + O(\rho)). 
\end{equation}

Finally, there is also a famous four-parameter ambiguity \cite{beig_einsteins_1982,Chrusciel:1989ye} in the differential structure at $\mathcal{I}^0$/$\mathcal{I}^\pm$, pertaining to the fact that changing the boundary defining function according to:
\begin{equation}
   \rho \quad \mapsto \quad \rho(1 -H \rho\log(\rho) + O(\rho)),
\end{equation}
where $(D_{\alpha}D_{\beta} +\nu_0 h_{\alpha \beta} )H=0$, will give another asymptotically flat spacetime. These log-symmetries will play no role in the present article.

\subsection{Asymptotic symmetries}

The algebra of asymptotic symmetries of \eqref{Ashtekar-Romano: metric Coordinates expansion} is given by vector fields $\xi = \xi^{\rho} \partial_{\rho} + \xi^{\alpha} \partial_{\alpha}$ whose Lie derivative preserves the form of the physical metric, including the values of $\nu_0$ and $h_{\alpha \beta}$. In this case, we have\footnote{Here, we are omitting potential log-symmetries; they play no role in the present work.} \cite{Ashtekar:1991vb,Compere:2023qoa}:
\begin{align}\label{Ashtekar-Romano: spi algebra}
    \xi^{\rho} &= -\rho^2 \omega(y) + o(\rho^2),  &
    \xi^{\alpha} &= \chi^{\alpha}(y) + \rho \nu_0^{-1}\partial^\alpha\omega(y) + o(\rho),
\end{align}
where $\chi^{\alpha}$ is a Killing vector for the hyperbolic metric $h_{\alpha \beta}$ on $\mathcal{I}^{0}$/ $\mathcal{I}^{\pm}$, $\mathcal{L}_{\chi}h_{\alpha \beta}=0$. Altogether, these vector fields form the algebra of the $\SPI$ group\footnote{Note that, since a generic Einstein metric $h_{\alpha\beta}$ need not to have any isometries, this really makes use of the assumption that $\mathcal{I}^{0}$/ $\mathcal{I}^{\pm}$ is $dS_n$/ $H^n$} \cite{Ashtekar:1991vb}:
\begin{equation}
    \SPI_{n+1} \;\simeq\; \textrm{SO}(1,n) \ltimes C^{\infty}\correc{\left(\mathcal{I}\right)}.
\end{equation}
It will be of importance for us that asymptotic symmetries \eqref{Ashtekar-Romano: spi algebra} act infinitesimally on the tensor $k_{\alpha\beta}(y)$ in the metric \eqref{Ashtekar-Romano: Coordinates expansion} according to:
\begin{align}\label{Ashtekar Romano: transformation of k}
    k_{\alpha \beta}(y) \quad\mapsto\quad k_{\alpha \beta}(y)  + \mathcal{L}_{\chi}k_{\alpha \beta}(y) + 2 \left( h_{\alpha \beta} + \nu_0^{-1} D_\alpha D_\beta \right)\omega(y).
\end{align}

\subsection{Projective compactness}\label{ssection:Projective compactness}

For comparison with the work of Ashtekar--Hansen, it will be enlightening to relate the statement of Definition~\ref{Definition:AsymptoticFlatness} to the notion of projective compactness (of order 1) as introduced by Čap and Gover in~\cite{Cap:2014ab}. 

This section is independent from the rest of the article and can safely be skipped upon first reading.

\begin{definition}
Let $M$ be a manifold with boundary with interior $\hat{M}$ and boundary $\partial M$. A metric $\hat{g}$ on $\hat{M}$ is said to be projectively compact (of order $1$) if near every point $x_0\in \partial M$ one can find a neighbourhood $U$ of $x_0$ and a boundary defining function $\rho: U\rightarrow \R_+$ such that the connection $\nabla$ on $U\cap \hat{M}$  constructed from the Levi-Civita connection $\hat\nabla$  by: 
\begin{equation}\label{eq:ProjectivelyCompactConnection}
\nabla_\eta \xi^a = \hat\nabla_\eta \xi^a + \frac{\eta^c\nabla_c\rho}{\rho}\xi^a + \frac{\xi^c\nabla_c\rho}{\rho}\eta^a, \end{equation}
has a smooth extension to $U\cap \pa M$. 
\end{definition}
 In practice, one checks the above condition by considering the Christoffel symbols in a local boundary chart near $x_0$ which are related to those of $\hat{\nabla}$ by
\begin{equation}\label{eq:ProjectivelyCompactConnection, christofel}
    \Gamma^i_{jk}=\hat\Gamma^{i}_{jk} + \frac{2}{\rho}\delta^i_{(j}\nabla_{k)}\rho.
\end{equation}

Equation~\eqref{eq:ProjectivelyCompactConnection} is precisely the condition that $\nabla$ and $\hat{\nabla}$ are \emph{projectively equivalent}, that is, they have the same \emph{unparametrised} geodesics. The following proposition gives the precise sense in which boundary points can be thought of as end points of geodesics of $\hat{g}$.
\begin{proposition}[{From \cite{Cap:2014aa}, see Proposition 2.4}] Let $\hat\nabla$ be the Levi-Civita connection of a projectively compact metric $\hat{g}$ on $\hat{M}$ (of order 1). Let $x_0\in\partial M$ and suppose that there is a geodesic curve  $\gamma : [0, 1] \rightarrow M$ for the connection \eqref{eq:ProjectivelyCompactConnection} which reaches $x_0$ transversely i.e. $\gamma(1)=x_0$, $\dot{\gamma}(1) \notin T\partial M$. Then there exists an asymptotic parametrisation $\hat{\gamma} : [0, \infty) \rightarrow M$ of the curve which satisfies $\hat{\gamma}([0, \infty)) \subset \hat M$,  $\displaystyle \lim_{t\to \infty} \gamma(t) = x_0$ and such that $\gamma$ is an affinely parameterised geodesic for $\hat\nabla$.\end{proposition}

A necessary condition for projective compactness of an Einstein metric is the following.
\begin{proposition}[{From \cite{Cap:2014aa,Cap:2014ab}}]
  Let ($\hat{M}$,$\hat{g}$) be a projectively compact Einstein metric with vanishing cosmological constant, then it must be asymptotically flat at time/space-like infinity in the sense of Definition \ref{Definition:AsymptoticFlatness}.
\end{proposition}
That the converse statement is not true in general can be seen from computing the Levi-Civita connection of \eqref{Ashtekar-Romano: metric Coordinates expansion}
\begin{equation}
    \hat\Gamma^{i}_{jk} = - \frac{2}{\rho}\delta^i_{(j}\nabla_{k)}\rho - \frac{\nu_0}{\rho} \left( g^{il}\nabla_l  \sigma\, \nabla_j\rho \nabla_k \rho \right)  + O(\rho^0).
\end{equation}
Thus, a metric that is asymptotically flat at time/space-like infinity in the sense of Definition \ref{Definition:AsymptoticFlatness} is projectively compact if and only if the mass aspect $\sigma$ is constant along the boundary $\partial_{\alpha}\sigma=0$. However, by \eqref{Ashtekar Romano: Einstein equations on k}, this means that $\sigma$ vanishes.

 Allowing a non-vanishing mass aspect therefore necessarily breaks projective compactness, leading to $O(1/\rho)$ terms in the connection $\nabla$ defined by~\eqref{eq:ProjectivelyCompactConnection}. Nevertheless, one can observe that projective compactness still \emph{almost} holds:  the divergent terms are located in the transverse components of the connection forms. Hence, although the projective structure does not extend to the boundary, the general picture of projective compactness remains intact. Importantly, this is true with regards to the geometric structures one can obtain at the boundary, therefore projective geometry still provides insight and intuition for studying asymptotic flatness at space/time-like infinity.

 \newpage
\section{Ti and Spi}\label{Section: Ti and Spi}

In this section, we introduce our definition of $\Ti$ and $\Spi$,
\begin{align}
    \Ti^{\pm} \simeq \mathcal{I}^{\pm} \times \mathbb{R} &\to \mathcal{I}^{\pm}, & \Spi \simeq \mathcal{I}^{0} \times \mathbb{R} &\to \mathcal{I}^{0},
\end{align}
as canonical extensions of the boundary $\mathcal{I}^{\pm}$/ $\mathcal{I}^{0}$ of asymptotically flat spacetime at time/spatial infinity (in the sense of Ashtekar--Romano). As we will prove later, in the flat case, our definition will be isomorphic to the homogeneous model of \cite{Figueroa-OFarrill:2021sxz}.  When the comparison makes sense, it will also be equivalent to the construction of Ashtekar--Hansen \cite{ashtekar_unified_1978}.

After introducing these extended boundaries, we explain why they are the natural places for asymptotic data of massive fields and how asymptotic symmetries are canonically realised as the group of automorphisms.

\subsection{Definition of Ti and Spi}\label{Subsection: definition Ti/Spi}

As we have already emphasised, the boundaries
\begin{equation}
    \mathcal{I}\in \{\mathcal{I}^{\pm}, \mathcal{I}^{0}\} \underset{i}{\hookrightarrow} M
\end{equation}
come with a preferred boundary defining function $\rho$ defined up to $\rho \mapsto \rho + O(\rho^2)$. At the level of jets, this defines a section
 \begin{equation}\label{Jlambda section}
 	\phi:\;  \mathcal{I} \xrightarrow{J^2 \rho}  \frac{i^*\left(J^2 M\right)}{ i^* \left(J^2 \rho^2\right)},
 \end{equation}
 of the canonical projection $i^*(J^2M) \rightarrow \frac{i^*(J^2M)}{i^*(J^2\rho^2)}$.
 
On the other hand, there is no canonical way to make a choice of $\rho$ in a neighbourhood of $\mathcal{I}$ in $M$ and thus no preferred way to uniquely extend the section $\phi$ to a section of $i^*\left(J^2M\right) \to \mathcal{I}$; this is the basis of the definition below for $\Ti$/$\Spi$. We will make use of the fact that 
\begin{equation}
   L := i^*\left(J^2M\right) \to  \frac{i^*\left(J^2M\right)}{ i^* \left(J^2 \rho^2\right)}
\end{equation}
is a line bundle and use the map \eqref{Jlambda section} to pull back this line over $\mathcal{I}$.

\begin{definition}\label{Definition: Spi}\mbox{}
The line bundle $\Spi \to \mathcal{I}^{0}$ which we refer to as ``spi'' is defined as the pullback bundle
\begin{equation}
\phi^* L \longrightarrow \mathcal{I}^{0}.
\end{equation}
This is summarised by the diagram:
\begin{center}
\begin{tikzcd}
\Spi=\phi^*L \arrow[d] & L=i^*\left(J^2M \right)\arrow[d]\\
\mathcal{I}^{0} \arrow[r,"\phi"]&\frac{i^*\left(J^2M \right)}{ i^* \left(J^2 \rho^2\right)}
\end{tikzcd}
\end{center}
\end{definition}
The definition is identical for ``past/future tie'' $\Ti^{\pm} \to \mathcal{I}^{\pm}$:
\begin{definition}\label{Definition: Ti}\mbox{}
The line bundles $\Ti^{\pm} \to \mathcal{I}^{\pm}$ are defined as the pullback bundles
\begin{equation}
\phi^* L \longrightarrow \mathcal{I}^{\pm}.
\end{equation}
\end{definition}

The definition should be read as follows: a point $(\rho_u, \bm{y})$ in $\Ti$ or $\Spi$, with coordinates $(u,\y) \in \mathbb{R}\times \mathcal{I}$, consists in a point $\bm{y}$ in $\mathcal{I}$ together with a choice of $2$-jet $\rho_u \in J^2_{\y}M$ of the form $\rho_u = j^2_{\y} \rho + u  j^2_{\y} \rho^2$. In practice however, when performing computations, we will often abuse notation and extend this $2$-jet to a boundary defining function $\rho_u \in \mathcal{C}^{\infty}(M)$
\begin{equation}\label{eq:ConventionsCoordinateu}
   \rho_u = \rho \Big(1 + \rho\, u + O(\rho^2) \Big),
\end{equation}
and check a posteriori that our results do not depend on the choice of extension. In particular, no derivatives in $u$ should appear.

\subsection[Application: Boundary value of a massive scalar field at Ti]{Application: Boundary value of a massive scalar field at $\Ti^+$}

As an example of application of this definition, let us study how massive scalar fields induce a boundary value on $\Ti^{+}$ in flat space.

Let
\begin{equation}
    \bm{k} \in \mathbb{R}^n \mapsto  \begin{pmatrix}
        \sqrt{1 + \bm{k}^2}\\ \bm{k}
    \end{pmatrix} \in \mathbb{R}^{1,n}
\end{equation}
be coordinates on the future unit hyperboloïd $H^n \subset \mathbb{R}^{1,n}$. Consequently, a moment of mass $m$ is parametrised by $\bm{k} \mapsto P^{\mu} = m \tilde{P}^{\mu}(\bm{k})$ with
\begin{equation}\label{Massive scalar field: P tilde definition}
\tilde{P}^{\mu}(\bm{k}) :=  \begin{pmatrix}
        \sqrt{1 + \bm{k}^2}\\ \bm{k}
    \end{pmatrix}.   
\end{equation}
A real valued solution $\Phi$ of the massive scalar field equation (with mass $m$) in $d=n+1$ dimensions can then be written in integral form:
\begin{equation}\label{Massive scalar field: Fourier coefficients}
    \Phi(X) = \frac{m^{n-1}}{2 (2\pi)^n}\int \frac{d^{n}\bm{k}}{\sqrt{1 + \bm{k}^2}} \Big( a(\bm{k}) e^{i m \tilde{P} \scal X} + a(\bm{k})^{\dagger} e^{-i m \tilde{P} \scal X}\Big).
\end{equation}

Introducing Beig-Schmidt type of coordinates $(\rho, \bm{y})\in \mathbb{R}^+\times \mathbb{R}^{n}$ on the inside of the future light-cone $\mathbb{M}^+ \subset \mathbb{M}^{1,n}$,
\begin{equation}\label{Massive scalar field: BS coordinate}
    X^{\mu} = \rho^{-1} \begin{pmatrix}
        \sqrt{1 + \bm{y}^2}\\ \bm{y}
    \end{pmatrix},
\end{equation}
the standard flat-space metric can then be written:
\begin{equation}\label{Massive field: flat metric expression}
    \pg_{ab} = - \rho^{-4} d \rho^2 + \rho^{-2} h_{ab}, \qquad h_{ab} = \left(\delta_{\alpha\beta} - \frac{y_{\alpha}y_{\beta}}{ 1+ |y|^2}\right) dy^{\alpha}dy^{\beta}.
\end{equation} 
 The asymptotic behaviour of the field at (future) time-like infinity is then given by \cite{JEDP_1987____A1_0},\cite[Chapter 7]{Hormander:1997aa}, see also \cite{Campiglia:2015kxa,Have:2024dff},\footnote{There appears to be some discrepancy in the literature regarding the phase factor $e^{in\frac{\pi}{4}}$ in this expression. For this reason, we include details on the computation in Appendix \ref{Section::AppendixSPA}.}
\begin{equation}
    \Phi(X)   \quad \underset{\rho \to 0}{\sim} \quad  \rho^{\frac{n}{2}}  \, \frac{m^{\frac{n}{2}-1}}{2 (2\pi)^{\frac{n}{2}}} \left( a(\y) e^{-im \rho^{-1}-in\frac{\pi}{4}} + a(\y)^{\dagger} e^{im \rho^{-1}+in\frac{\pi}{4}} \right)  + O(\rho^{\frac{n}{2}+1} ).
    \label{Eq::FallOffScalar}
\end{equation}

We can use this to derive the value of the field at $\Ti^{+}$, where ``plus" here stands for ``future".
\begin{proposition}\mbox{}

\label{Proposition::MassiveScalar}

Let $(u, \y )$ be a point of $\Ti^{+}$, then the limit
\begin{align}
    \varphi(u, \y) &:= \lim_{\rho\to 0} \quad \frac{1}{2}\Big(    \Phi(X) - \frac{1}{i m} \nabla^a \rho_u \nabla_a\Phi(X) \Big) \rho_u^{-\frac{n}{2}}e^{in\frac{\pi}{4}} e^{i m \rho_u^{-1}} \label{Eq::AsymptoticScalar1}
\end{align}
exists and defines a field on $\Ti^+$. In terms of the Fourier coefficients \eqref{Massive scalar field: Fourier coefficients}, this reads
\begin{equation}
    \varphi(u, \y) = \, \frac{m^{\frac{n}{2}-1}}{2 (2\pi)^{\frac{n}{2}}}  a(\y) e^{-im u}. \label{Eq::AsymptoticScalar2}
\end{equation}
We will call this function the scattering data induced by the massive field on $\Ti$.
\end{proposition}
\begin{proof}\mbox{}
By hypothesis, $\nabla^a \rho_u = \rho^{2}\left( - \partial_{\rho} + h^{\alpha\beta}\nabla_{\alpha} u \,\partial_{\beta} \right) +O(\rho^3)$; note that the index has been raised using the rescaled inverse metric $g^{ab}$. Consequently,
\begin{equation}
    - \frac{1}{i m}\nabla^a \rho_u \nabla_a\Phi(X) \underset{\rho \to 0}{\sim} \rho^{\frac{n}{2}} \, \frac{m^{\frac{n}{2}-1}}{2 (2\pi)^{\frac{n}{2}}}  \left( a(\y) e^{- im \rho^{-1}-in\frac{\pi}{4}} - \, a(\y)^{\dagger} e^{im \rho^{-1}+in\frac{\pi}{4}} \right)  + O(\rho^{\frac{n}{2}+1} )
\end{equation}
from which we obtain
\begin{equation}
\frac{1}{2}\left(\Phi(X) - \frac{1}{i m}\nabla^a \rho_u \nabla_a\Phi(X) \right) \underset{\rho \to 0}{\sim} \rho^{\frac{n}{2}} \, \frac{m^{\frac{n}{2}-1}}{2 (2\pi)^{\frac{n}{2}}}  \, a(\y) e^{- im \rho^{-1}-in\frac{\pi}{4}}    + O(\rho^{\frac{n}{2}+1} ).
\end{equation}
This selects one of the asymptotic plane waves, and the remaining factors lift it to a field on $\Ti^+$:
\begin{equation}
\frac{1}{2}\left(\Phi(X) - \frac{1}{i m}\nabla^a \rho_u \nabla_a\Phi(X) \right)  (0^+-i\rho_u^{-\frac{n}{2}})  e^{i m  \rho_u^{-1}} \underset{\rho \to 0}{\sim}  \frac{m^{\frac{n}{2}-1}}{2 (2\pi)^{\frac{n}{2}}} \, a(\y) e^{- im u}    + O(\rho).
\end{equation}

\end{proof}
Note that in the final result no derivatives of $u$ appear, indicating that the limit only depends on the $2$-jet of $\rho_u$,  which is the necessary condition for it to genuinely define a field on $\Ti^+$.

\subsection{Carrollian geometry and Asymptotic symmetries}

The Carroll limit was first studied by Levy-Leblond \cite{levy-leblond_nouvelle_1965} (and independently by Sen Gupta \cite{SenGupta:1966qer}) as an alternative to the more usual Galilean limit. Both limits can be realised as Inönü-Wigner contractions of the Poincaré group but the geometry (and physics) of the resulting spaces is obviously very different, we refer to \cite{duval_carroll_2014} for a wonderful exposition of the subject. As was emphasised in \cite{duval_conformal_2014} the BMS group can be realised as the group of conformal automorphisms of a (weak) Carrollian geometry; see also \cite{geroch_asymptotic_1977,ashtekar_geometry_2015} for classical work relating the BMS group to the null geometry of \enquote{scri} $\mathscr{I}$.

At spatial infinity, where asymptotic symmetries form the SPI group of \cite{ashtekar_unified_1978,Ashtekar:1991vb}, the situation might appear to be quite different at first glance. However, as was remarked in \cite{gibbons_ashtekar-hansen_2019}, these symmetries can still be understood as the Carroll automorphisms of Ashtekar-Hansen's boundary at spatial infinity. As we mentioned previously, and will be discussed again in further detail in Section \ref{ssection: Comparison with the previous literature}, when Definition \ref{Definition: Spi} and Ashtekar-Hansen's definition both make sense, they essentially coincide. It is therefore not surprising that the SPI group will be realised a group of Carrollian automorphisms of $\Spi$ and $\Ti$. This result will in particular lead to a realisation of a representation of the SPI group on massive fields. 

As was first pointed out in \cite{troessaert_bms4_2018}, the BMS group can also be recovered as a preferred subgroup of asymptotic symmetries at spatial infinity. We will postpone discussion of this point to Section \ref{Section: Carrollian geometry of Spi}, where we will show that $\Spi$ and $\Ti$ are actually examples of strongly Carrollian geometries.

\subsubsection{Carrollian geometry}\label{sssection: Carrollian geometry}

\begin{proposition}\mbox{} \label{Proposition: principal bundle}

$\Ti/\Spi$ are naturally $\mathbb{R}$-principal bundles over $\mathcal{I}^{\pm}/\mathcal{I}^{0}$. In other terms, we have a canonical complete vector field $n^{a}$ along the fibres.
 
The natural coordinate system $(u,\bm{y})$ previously described, provides a natural global trivialisation: the $\mathbb{R}$-action is a just constant translation in $u$ and $n^{a} = \partial_u$.
\end{proposition}
\begin{proof}\mbox{}
Let $\rho_u$ be an arbitrary admissible boundary defining function. Then, referring to the discussion in section~\ref{Section: Ashtekar-Romano spacetimes}, $(\rho_u)'=\rho_u + \alpha \rho_u^2$ is also an admissible boundary defining function and its $2$-jet at any boundary point defines a point in $\Ti/\Spi$ that only depends on the $2$-jet of $\rho_u$ at that point, i.e. the point of $\Ti/\Spi$ defined by $\rho_u$. For each $\alpha \in \R$, we obtain in this way a diffeomorphism $R_\alpha$ of the extended boundary. In the canonical coordinates $(u,\bm{y})$ given by \eqref{eq:ConventionsCoordinateu} is simply given by:
\[ R_\alpha: (u,\bm{y}) \mapsto (u+\alpha, \bm{y}).\]
Note that the $\mathbb{R}$-action on the 2-jets at $\y$ can also be concisely be written as $R_\alpha (\rho_u) = \rho_u + \alpha \,j^2_{\y} \rho^2$. 
\end{proof}

In particular, one sees from \eqref{Eq::AsymptoticScalar2} that the scattering data $\varphi$ induced on $\Ti$ by a massive scalar field must satisfy $\big(\mathcal{L}_n + i m \big) \varphi =0$. These are in fact equivariant functions under the $\mathbb{R}$-action, 
\begin{equation}
    R_{\alpha}^*\varphi(u,\y) = \varphi(u+\alpha,\y) = e^{-im \alpha}\varphi(u,\y) = \rho_m(-\alpha) . \varphi(u,\y)
\end{equation}
for the representation 
\begin{equation}
    \rho_{m} \left|\begin{array}{ccc}
       \mathbb{C} \times \mathbb{R}  & \mapsto &  \mathbb{C}  \\
        (\varphi , \alpha)  &\mapsto & e^{i m \alpha} \varphi
    \end{array}\right..
\end{equation}
We therefore have the following:
 \begin{proposition}\label{Proposition: scattering data are sections}
     The space of scattering data $\varphi(u,\y)$ of massive fields of mass $m$ on $\Ti^+$ is in to one-to-one correspondence with the space of sections $\y \mapsto \underline{a}(\y)$ of the associated bundle $\Ti^+ \times_{\rho_m} \mathbb{C} \to H^n$ over $H^n$.
 \end{proposition}

 \begin{proof}
     By definition, a section of the associated bundle $\Ti^+ \times_{\rho_m} \mathbb{C} \to H^n$ over $H^n$ is given by
 \begin{equation}
    \underline{a}(\y) \;\left|\quad  \begin{array}{ccc}
       H^n  & \to & \Ti^+ \times_{\rho_m} \mathbb{C} \\[0.4em]
       \bm{y}  &  \;\mapsto\; & \Big[  \big( u \;,\; \bm{y} \big)  \;,\; \varphi(u, \y) \Big]
    \end{array}\right.,
\end{equation}
where brackets denote the equivalence class for the relation, $\big( ( u , \bm{y} ) ,  \varphi \big) \sim \big( ( u + \alpha , \bm{y} ) ,  \rho_{m}(-\alpha). \varphi \big)$, defining $\Ti^+ \times_{\rho_m} \mathbb{C}$.
 \end{proof}

It follows from our definition and Proposition \ref{Proposition: principal bundle} that $\Spi/\Ti$ are naturally equipped with a (weakly) Carrollian geometry $(h_{ab}, n^{a})$, consisting of a vector field and degenerate metric (of signature $(0,\mp,+ \dots +)$ such that
\begin{align}
    n^a h_{ab}&=0, & \mathcal{L}_n h_{ab}&=0.
\end{align}
In our preferred local coordinates $(u,\y)$, these are $n^a= \partial_u$ and $h_{ab} = h_{\alpha \beta }(\y) dy^{\alpha}dy^{\beta}$ where $ h_{\alpha \beta }(\y)$ is the hyperbolic metric on $\mathcal{I}^{0}/\mathcal{I}^{\pm}$.

\subsubsection{The SPI group}

As was already noted in \cite{RSTA20230042}, asymptotic symmetries of the spacetime coincides with automorphisms of the Carrollian geometry on  $\Spi/\Ti$. This parallels the situation at null infinity where the asymptotic symmetry group reduces along null infinity to conformal Carrollian symmetries \cite{geroch_asymptotic_1977,duval_conformal_2014}.

\begin{proposition}\label{Proposition: Spi group as automorphism}
    The asymptotic algebra symmetries \eqref{Ashtekar-Romano: spi algebra} naturally acts on $\Spi/\Ti$. Moreover, this action coincides with the algebra automorphisms of the Carrollian structure $(h_{ab}, n^{a})$, i.e. vector fields $\bar\xi$ satisfying $\mathcal{L}_{\bar\xi} h_{ab} = \mathcal{L}_{\bar\xi} n^a = 0$:
     \begin{equation}\label{Ti: spi algebra}
         \bar\xi(y) = \omega(y)\partial_u + \chi^\alpha(y) \partial_{y^\alpha}.
     \end{equation}

\end{proposition}
\begin{proof}\mbox{}
Let $\Phi_t^\xi$ denote the local 1-parameter group of diffeomorphisms generated by the vector field $\xi$ given by Eq.~\eqref{Ashtekar-Romano: spi algebra}. Define a local $1$-parameter group of diffeomorphisms on $J^2M$ by:
\[ j^2f_x \mapsto j^2(f\circ\Phi_{-t}^\xi)_{\Phi_t^\xi(x)}. \] 

Explicitly, in canonical coordinates on $J^2M$ i.e. such that:
\[ j^2f_p = (\rho(p), y^\alpha(p), f(p), \partial_\rho f(p), \partial_\alpha f(p),\begin{bmatrix} \partial^2_\rho f(p) & \partial_\rho\partial_\alpha f(p) \\  \partial_\rho\partial_\alpha f(p) & \partial_\alpha\partial_\beta f(p)\end{bmatrix}),\]
a point $(u,y)$ of the extended boundary can be identified with a $2$-jet:
\[ (0,y^\alpha, 0,1,0,\begin{bmatrix}2u&0\\0&0\end{bmatrix}).\]
The image of which under the local diffeomorphism is given explicitly by:
\begin{equation}\label{eq:InducedSymmetry} (0, y^\alpha(\Phi_{t}(p)), 0, \partial_\rho\Phi^\rho_{-t}(\Phi_t(p)), \partial_\alpha \Phi^\rho_{-t}(\Phi_t(p)),[2u\partial_i\Phi^\rho_{-t}(\Phi_{t}(p))\partial_j\Phi^\rho_{-t}(\Phi_{t}(p)) + \partial_i\partial_j \Phi^\rho_{-t}(\Phi_t(p))]),\end{equation}
taking the derivative with respect to $t$ evaluated at $t=0$ we find that at these points the infinitesimal generator is:
\[(0, \chi^\alpha(y), 0, 0, 0, \begin{bmatrix}-\partial^2_\rho \xi^\rho (p) & 0 \\ 0 & 0 \end{bmatrix})=(0,\chi^\alpha(y),0,0,0,\begin{bmatrix}2\omega(y) & 0 \\ 0 & 0\end{bmatrix}).\]
this shows that the restriction of the diffeomorphism to these jets yields a diffeomorphism on the extended boundary generated by the vector field: 
\[ \bar\xi(y) = \omega(y)\partial_u + \chi^\alpha(y) \partial_{y^\alpha}. \]
\end{proof}

Combining this realisation of the $\SPI$ group with the expression \eqref{Eq::AsymptoticScalar2} for massive fields we obtain the following.
\begin{proposition}\label{Proposition: Action of the SPI group on massive fields}
   The action of the $\SPI$ group on the scattering data \eqref{Eq::AsymptoticScalar2} of a massive scalar field is given by
   \begin{align}\label{Action of the SPI group on massive fields}
      \varphi\big(u, \y\big) &\;\mapsto   \varphi\big(u +\omega(\y), \y'(\y)\big) && \Longleftrightarrow &  a(\y) &\;\mapsto a(\y') e^{-im \omega(\y)}.
   \end{align}
\end{proposition}
It follows from Proposition \ref{Proposition: Action of the SPI group on massive fields} that scattering data of massive fields form a unitary irreducible representation (UIR) of the $\SPI$ group. The norm preserved by the representation coincides with the usual norm on the space of massive fields,
\correc{\begin{align*}
    ||\varphi||^2 &= \frac{m^{n-1}}{2 (2\pi)^{n}} \int_{H} \frac{d^{n}\y}{\sqrt{1 + \y ^2}}\, a(\y) a^*(\y).
\end{align*}}
 Therefore, this shows that usual massive UIR of the Poincaré group can naturally be extended to UIR of the SPI group. This is similar to the situation of massless fields at null infinity: there, as was highlighted in \cite{Bekaert:2022ipg}, the massless UIR of the Poincaré group lifts to an UIR of the BMS group (Sachs representation \cite{Sachs:1962zza}).
 
 As we shall discuss in Section \ref{ssection: BMS group at Ti/Spi}, restricting the action of Proposition \ref{Proposition: Action of the SPI group on massive fields} to the BMS group will yield (see Proposition \ref{Proposition: Action of the BMS group on massive fields}) a geometrical realisation of Longhi–Materassi BMS representation \cite{Longhi:1997zt}. \correc{BMS UIRs have been classified by McCarthy in a series of papers~\cite{Mccarthy:1972ry,McCarthy_72-I,McCarthy_73-II,McCarthy_73-III,McCarthy_76-IV,McCarthy_75,McCarthy_78,McCarthy_78errata} (see also \cite{Girardello:1974sq}); from this perspective, the representation that will be obtained in Proposition \ref{Proposition: Action of the BMS group on massive fields} is very specific: see \cite{Bekaert:2024jxs,Bekaert:2025kjb} for a physical interpretation of generic BMS UIRs as well as a detailed discussion on how the Longhi–Materassi BMS representation sits in McCarthy's classification.}

\subsubsection{SPI versus BMS supertranslations}

The group of $\SPI$ supertranslations $\mathcal{T}_{\SPI} \simeq \mathcal{C}^{\infty}(\mathcal{I})$ is the abelian group of functions $\omega(\y)$ on $\mathcal{I}$. BMS supertranslations $\mathcal{T}_{\BMS} \simeq \mathcal{E}[1](S^{n-1})$ form the abelian group of weight one conformal densities $\xi(z^i)$ on the celestial sphere $S^{n-1}$.

Let $z^i \in \mathbb{R}^{n-1}$ be stereographic coordinates on the celestial sphere $S^{n-1}$ and let $z^i \mapsto \tilde{Q}^{\mu}\left(z^i\right)$ be the embedding of this sphere into the null cone of $\mathbb{R}^{n+1}$ given by
\begin{align}
    \tilde{Q}^{\mu}\left(z^i\right) = \begin{pmatrix}
        1+ z^2 \\ 2 z^i \\ 1-z^2
    \end{pmatrix}.
\end{align}

It is well known that this gives rise to a canonical embedding $\mathbb{R}^{n,1}  \hookrightarrow \mathcal{T}_{BMS}$ of translations into BMS supertranslations given by
\begin{equation*}
   \xi(z^i) = T^{\mu} \tilde{Q}_{\mu}\left(z^i\right).
\end{equation*}
Despite this, there is however no preferred embedding of the Poincare group inside the BMS group.

Similarly, at Ti, we have a canonical embedding $\mathcal{T}_{BMS}  \hookrightarrow \mathcal{T}_{\SPI}$ of BMS supertranslations into $\SPI$ supertranslations given by \cite{Campiglia:2015lxa,Campiglia:2015kxa}
\begin{equation}\label{Asympotic symmetries: BMS in SPI}
    \omega_\xi(\y) = \bigintsss_{S^{n-1}} (dz)^{n-1} \,\frac{ \xi\left(z^i\right)}{\left( \tilde{P}^{\mu}(\y) \tilde{Q}_{\mu}\left(z^i\right)\right)^{n}}
\end{equation}
where $\tilde{P}^{\mu}(\y)$ is given by \eqref{Massive scalar field: P tilde definition}. 
\begin{proof}
This is just the application of a bulk to boundary propagator for the equation $(D^2 + \nu_0 n) \omega(\y)=0$. For completeness, let us give a brief proof.
    By construction $\tilde{Q}_{\mu}$ is a map from $\mathbb{R}^{3,1}$ to $\mathcal{E}[1](S^{n-1})$ while $\tilde{P}_{\mu}$ is a map from $\mathbb{R}^{3,1}$ to $C^{\infty}(H^n)$. Since we are working in stereographic coordinates $(dz)^{n-1}$ coincides with the canonical volume form of conformal weight $n-1$ on the sphere. Therefore the integrand is a genuine, i.e. non weighted, volume form on the celestial sphere. Now since $\tilde{P}^{\mu}$ is timelike and $\tilde{Q}^{\mu}$ is null then $\tilde{P}^{\mu}(\y) \tilde{Q}_{\mu}$ is nowhere zero and the volume form is smooth. We conclude that the integral is finite and defines a function on $\mathcal{I}^{\pm} = H^n$. That \eqref{Asympotic symmetries: BMS in SPI} is always a solution $(D^2 + \nu_0 n) \omega(\y)=0$, can be checked using identities \eqref{Flat model: identity on P}.
\end{proof}
Here again, without introducing further geometrical data, there is no preferred embedding of the BMS group inside the $\SPI$ group. However, it turns out that Ti is always equipped with a Carrollian connection that will provide such a choice of BMS group.

Note that for $\Spi$ the situation is more complicated since the above propagator will have some divergences which then need to be regularised and leads to an extra branch of solution.

\subsection{Comparison with the previous literature}\label{ssection: Comparison with the previous literature}

\paragraph{Homogeneous models}\mbox{}

In \cite{Figueroa-OFarrill:2021sxz}, Figueroa-O'Farril--Have--Prohazka--Salzer introduced the names $\Ti$ and $\Spi$ for the (pseudo) Carrollian homogeneous spaces associated to the Poincaré group $ISO(1,n) = \textrm{SO}(n,1) \ltimes \mathbb{R}^{n+1}$:
\begin{align}\label{Ti and Spi: flat models}
\begin{alignedat}{2}
        \Ti^{n} = \frac{\textrm{SO}(n,1) \ltimes \mathbb{R}^{n+1}}{ \textrm{SO}(n) \ltimes \mathbb{R}^{n}} &\to  H^n = \frac{\textrm{SO}(n,1) }{ \textrm{SO}(n)},\\[0.5em]
        \Spi^{n} = \frac{\textrm{SO}(n,1) \ltimes \mathbb{R}^{n+1}}{ \textrm{SO}(n-1,1) \ltimes \mathbb{R}^{n}} &\to  dS_n = \frac{\textrm{SO}(n,1) }{ \textrm{SO}(n-1,1)}.
\end{alignedat}
\end{align}
We shall prove, see Section \ref{subsection: Poincaré action and the flat model of Ti}, that, in the flat case, Definition \ref{Definition: Ti} of $\Ti$ implies that this space is naturally equipped with a transitive action of the Poincaré group, induced by the action of the bulk, and is therefore canonically identified with the homogeneous space \eqref{Ti and Spi: flat models}. It should also be clear that, when adapted to spacelike infinity, the same results are true for $\Spi$. Furthermore, we shall also see (below) that Definition \ref{Definition: Spi} can be understood to coincide with that of Ashtekar--Hansen \cite{ashtekar_unified_1978} in the curved case and, therefore, that it bridges between the two approaches.

\paragraph{Ashtekar--Hansen construction}\mbox{}

In the seminal work \cite{ashtekar_unified_1978}, spatial infinity $\mathcal{I}^0=H^3$ was constructed as a blow up of $\iota^0$, but the authors also introduced a natural line bundle over $\mathcal{I}^0$. It was already remarked in \cite{gibbons_ashtekar-hansen_2019} that the resulting $4$d manifold is (weakly pseudo) Carrollian. As we shall now explain, this construction of Ashtekar--Hansen is essentially equivalent to Definition \ref{Definition: Spi}. The main difference is that Ashtekar--Hansen's definition is really tied up to assumptions about how future and past null infinity connect to each other while Definition \ref{Definition: Spi} relies on the existence on Ashtekar--Romano's asymptotics. Furthermore, as will be shown in Section \ref{Section: Carrollian geometry of Spi}, it follows quite straightforwardly from our definition of $\Ti/\Spi$ that they are \emph{strongly} Carrollian: the ``first order structure" of \cite{ashtekar_unified_1978,Ashtekar:1991vb}, in fact, corresponds to a Carrollian connection on $\Ti/\Spi$. This allows for a completely geometrical realisation of the BMS and Poincaré group at $\Ti / \Spi$

We now explain in which sense Definition \ref{Definition: Spi} can be understood to coincide with Ashtekar--Hansen's. In \cite{ashtekar_unified_1978}, the blow up $\mathcal{I}^0=H^{3}$ of $\iota^0$ is constructed as the space of endpoints of spatial geodesics $\gamma:\mathbb{R} \to \mathbb{M}$, $|\dot{\gamma}|>0$. This is, in essence, the idea underlying projective compactification, in the sense of Čap--Gover \cite{Cap:2014aa,Cap:2014ab} (see the discussion in Section \ref{ssection:Projective compactness}). As we recalled, projectively compact manifolds are asymptotically flat in the sense of Definition \ref{Definition:AsymptoticFlatness}; in the absence of mass aspect the converse is also true and end points of geodesics are then genuinely identified with $\mathcal{I}^{\epsilon}$. Asymptotically flat manifolds with non zero mass aspect are however not, strictly speaking, projectively compact and the identification in this case is more heuristic. 

Ashtekar--Hansen then define the line bundle attached to $\mathcal{I}^0$ by making the following remark: the norm of the asymptotic speed $|\dot\gamma|$ of the geodesics ending at $\iota^0$ can be fixed uniquely to be some fixed constant value, but this is not the case of the asymptotic tangential acceleration $\ddot{\gamma}\cdot \dot{\gamma}$. Points in the fibres over $\mathcal{I}^0$ were then defined to be a choice of such asymptotic tangential acceleration. This can be rephrased by saying that geodesics have an asymptotic preferred parametrisation $\rho$, with $\rho=0$ being the boundary, defined up to $\rho \mapsto \rho + u \rho^2 + O(\rho^3)$ and that point in the fibres are identified with choice of $u$. This should make it clear that, when the spacetime is projectively compact at least, Definition \ref{Definition: Spi} coincides with the extended boundary of Ashtekar--Hansen. This is in particular is the case at time-like infinity in the absence of mass or at spatial infinity for Minkowski space.

\section{An integral formula on Ti for massive fields on Minkowski space}\label{Section: An integral formula on Ti}

In this section, we will present an integral formula at $\Ti$ which allows to produce solutions of the massive Klein-Gordon equation,
\begin{equation*}
    \left( \Box - m^2 \right) \Phi(X) = 0.
\end{equation*}

It will proved a complete parallel of the Kirchhoff-d'Adhémar formula at null infinity \cite{Kirchhoff:1883aa,Penrose_1980,penrose_spinors_1984} for massless particles.

We first show that any point $X$ in Minkowski space uniquely defines a cut $\xi_X : H^n \to \Ti$, i.e. a section of $\Ti$, $\y \to (\y , u = \xi_X(\y))$. Secondly, we prove that integration of suitable functions $\varphi(u,\y)$ on these cuts produces solutions $\Phi(X)$ of the massive Klein-Gordon equation.

\subsection{Cuts of Ti}

Let us introduce again Beig-Schmidt coordinates  $(\rho, \bm{y})$ on the inside of the future light-cone in Minkoswki $\mathbb{M}^{+} \subset \mathbb{M}^{1,n}$, by 
\begin{equation}\label{cuts: beig-schmidt chart}
    \begin{array}{ccc}
        \mathbb{R} \times \mathbb{R}^n  & \to & \mathbb{M}^{+}  \\[0.2em]
         (\rho, \bm{y}) & \mapsto  & \rho^{-1} \begin{pmatrix}
        \sqrt{1 + \bm{y}^2}\\ \bm{y}
    \end{pmatrix}
    \end{array}.
\end{equation}
Expressed in these coordinates, the flat metric is then given by \eqref{Massive field: flat metric expression}.

We now pick a point $X^{\mu} \in \mathbb{M}^{1,n}$ of Minkowski space and a time-like future directed unit vector
\begin{equation}\label{cuts: beig-schmidt chart on momenta}
  \tilde{P}^{\mu}(\bm{k}) = \begin{pmatrix}
        \sqrt{1 + \bm{k}^2}\\ \bm{k}
    \end{pmatrix},
\end{equation}
 and consider the map, describing a geodesic shot from this point in the direction $\tilde{P}^{\mu}(\bm{k})$, given by:
\begin{equation}\label{cuts: flow chart}
    \begin{array}{ccc}
        \mathbb{R} \times \mathbb{R}^n  & \to & \mathbb{M}^{1,n}  \\[0.2em]
         (s, \bm{k}) & \mapsto  & X^{\mu} + s\; \tilde{P}^{\mu}(\bm{k})
    \end{array}.
\end{equation}
Expressing the flow \eqref{cuts: flow chart} in our coordinates \eqref{cuts: beig-schmidt chart}, we arrive at:
\begin{align}
\badat{2}
    \rho(s, \bm{k})  &= s^{-1} + s^{-2} X \scal \tilde{P}(\bm{k}) + O(s^{-3}) \\
    \bm{y}(s, \bm{k}) &= \bm{k} + O(s^{-1}).
    \eadat
\end{align}
As there is no mass aspect in Minkowski space-time, projective compactness implies that geodesics really extend to infinity (see Section \ref{ssection:Projective compactness}), and in a suitable neighbourhood of future time-like infinity \eqref{cuts: flow chart} actually defines a local chart. This chart can be compared to \eqref{cuts: beig-schmidt chart} to obtain the corresponding transition functions. Asymptotically, one has
\begin{align}
\badat{2}
    s^{-1}  &= \rho - \rho^2 X \scal \tilde{P}(\bm{y}) + O(\rho^{-3}) \\
    \bm{k}  &= \bm{y} + O(\rho).
    \eadat
\end{align}

We have therefore found that any point $X^{\mu}$ of Minkowski space generates a boundary defining function on time-like infinity, obtained by taking the inverse arc length $\rho_X := s^{-1}$ of the geodesic flow starting at $X^{\mu}$:
\begin{equation}
    \rho_X = \rho -\rho^2 X \scal \tilde{P}(\bm{y}) + O(\rho^{-3}).
\end{equation}
It associates a $2-$jet to each point $\y$ of the asymptotic hyperboloid and, consequently and by definition of $\Ti^{+}$, defines a section (or ``cut'') $\xi_X \in \Gamma\left[\Ti^+\right]$:
\begin{equation}\label{Cut of Ti: image of the cut}
    \xi_X \;\left|\quad  \begin{array}{ccc}
       H^n  & \to & \Ti^{+} \\[0.4em]
       \bm{y}  &  \;\mapsto\; & \Big( u=  -X \scal \tilde{P}(\bm{y}) \;,\; \bm{y}\Big)
    \end{array}\right. .
\end{equation}
A similar calculation for $\Ti^-$ would give exactly the same result (with no sign difference in the expression for the cuts). Note that, since we are working in mostly plus conventions for the metric and $\tilde{P}$ is future-directed, $\delta X \scal \tilde{P}(\bm{y})<0$ for any future-directed translation $\delta X$ and therefore the orientation of the $\mathbb{R}$ action at $\Ti$ coincides with the orientation induced by time translations.

\subsection{Reconstruction formula on Ti}

We now present a generalisation of the Kirchhoff-d'Adhémar formula for massive scalar fields. 

Recall from Proposition \ref{Proposition: scattering data are sections} that the scattering data $\varphi(u,\y)$ for massive fields of mass $m$ are equivariant functions on $\Ti$ for the representation $\rho_m$ of $(\R,+)$ given by $\rho_m(\alpha)\cdot \varphi = e^{i m \alpha} \varphi$; equivalently, these are functions on $\Ti$ that satisfy:
\begin{align}\label{Reconstruction formula: homogeneity condition}
        \big( \mathcal{L}_n + im \big)\varphi(u,\y) &=0 &\Leftrightarrow&& \varphi(u, \y) &= \, A(\y) e^{-im u}.
\end{align}

   The formula below will reconstruct the corresponding massive field. In parallel with the classical Kirchhoff formula, it involves, for any point $X$ of Minkowski, an integral $\int_{\xi_X}d\mu_{H}$ over the image of the corresponding cut \eqref{Cut of Ti: image of the cut}. Here $d\mu_{H}$ is the volume form on $H^n$, naturally pulled-back on the cut by the projection $\pi : \Ti \to H^n$.

\begin{proposition}\label{Proposition: Reconstruction formula}\mbox{}

    Given a section of $\Ti^+ \times_{\rho_m} \mathbb{C}$ or, equivalently, a function $\varphi: \Ti^+ \rightarrow \mathbb{C}$ satisfying \eqref{Reconstruction formula: homogeneity condition} then
    \begin{align}
        \Phi (X) = \left( \frac{m}{2 \pi} \right)^\frac{n}{2} \int_{\xi_X} d \mu_{H} \big( \varphi + \varphi^{\dagger}\big)
    \end{align}
satisfies the  Klein-Gordon equation, $(\Box - m^2)\Phi (X) =0$. Moreover, $\varphi(u,\y)$ can be recovered uniquely from $\Phi (X)$ through equation \eqref{Eq::AsymptoticScalar1}. From Proposition \ref{Proposition::MassiveScalar} all\footnote{\correc{A precise discussion of the functional spaces involved to make this statement precise is beyond the scope of this paper and is left for future work. As an initial step, one may restrict to solutions whose initial value is an element of the Schwartz space $\mathcal{S}(\R^3)$. }} solutions of the Klein-Gordon equations are obtained in that way.
\end{proposition}

\begin{proof}
The above proposition is coordinate free and therefore we can prove it in any coordinate system that we like. We will pick Beig-Schmidt coordinates in a neighbourhood of time-like infinity and prove that the proposed prescription recovers a field of the form \eqref{Massive scalar field: Fourier coefficients}.

Since $\varphi(u, \y)$ satisfies \eqref{Reconstruction formula: homogeneity condition} one can write $\varphi(u, \y) = \, A(\y) e^{-im u}$.
 The unit hyperboloid metric, $h_{ab}$
 in the Beig-Schmidt coordinates is given by \eqref{Massive field: flat metric expression} and the corresponding volume form is $\frac{d^n y}{\sqrt{1+\y^2}}$. Thus,
    \begin{align}
    \begin{alignedat}{3}
        & \left( \frac{m}{2 \pi} \right)^\frac{n}{2}\int_{\xi_X} d\mu_{H} \big( \varphi + \varphi^{\dagger} \big)\\
        &= \left( \frac{m}{2 \pi} \right)^\frac{n}{2} \int_{H^n} \left. \frac{d^n \y}{\sqrt{1 + \y^2}} \Big( A(\y) e^{- i m u}  + A(\y)^{\dagger} e^{i m u} \Big)\right\vert_{u = - \tilde{P}(\y) \scal X} \\ &= \left( \frac{m}{2 \pi} \right)^\frac{n}{2} \int \frac{d^n \bm{k}}{\sqrt{1 + \bm{k}^2}} \Big( A(\y) e^{i m \tilde{P}(\bm{k}) \scal X} +  A(\y)^{\dagger} e^{-i m \tilde{P}(\bm{k}) \scal X} \Big)
    \end{alignedat}
    \end{align}
which satisfies the massive Klein-Gordon equation. Comparing with \eqref{Massive scalar field: Fourier coefficients}, we find $A(\y) = \frac{m^{\frac{n}{2}-1}}{2 (2\pi)^{\frac{n}{2}}} a(\y)$ and thus $\varphi(u,\y) = \frac{m^{\frac{n}{2}-1}}{2 (2\pi)^{\frac{n}{2}}} a(\y) e^{-im u}$ which matches with \eqref{Eq::AsymptoticScalar1}.
\end{proof}

\subsection{Poincaré action and the flat model of Ti}\label{subsection: Poincaré action and the flat model of Ti}

It also follows from the previous discussion that an element of the Poincaré group $(m^{\mu}{}_{\nu}, T^{\mu}) \in \textrm{SO}(n,1) \ltimes \mathbb{R}^{n+1}$ acts on $\Ti$ according to
\begin{align*}
    \Big(u \;,\; \y^{\alpha}\Big) \quad \mapsto\quad  \Big(u - T\scal \tilde{P}(\y)\;,\; (\y')^{\alpha}\Big) \qquad \text{where} \quad \tilde{P}(\y')^{\mu} = m^{\mu}{}_{\nu}\tilde{P}(\y)^{\nu};
\end{align*}
see the detailed proof below. The action is transitive and the stabiliser of a given point $\Big(u \;,\; \y^{\alpha}\Big)$ is defined by the relations:
\begin{align*}
   T\scal \tilde{P}(\y) =0 \qquad  \tilde{P}(\y)^{\mu} = m^{\mu}{}_{\nu}\tilde{P}(\y)^{\nu},
\end{align*}
and is thus isomorphic to $\textrm{SO}(n) \ltimes \mathbb{R}^{n}$.  Therefore, $\Ti \to H^n$ coincides, in this situation, with the homogenous space
\begin{equation}
    \Ti^{(n)} = \frac{\textrm{SO}(n,1) \ltimes \mathbb{R}^{n+1}}{ \textrm{SO}(n) \ltimes \mathbb{R}^{n}} \to  H^n = \frac{\textrm{SO}(n,1) }{ \textrm{SO}(n)}.
\end{equation}
Hence we recover that, in the flat case, our definition of \eqref{Definition: Ti} matches the homogeneous model \eqref{Ti and Spi: flat models} discussed in \cite{Figueroa-OFarrill:2021sxz}.

\begin{proof}
Let us consider a Poincaré transformation $\phi : X^\mu \mapsto (X^\mu)' =T^\mu +m^\mu{}_\nu X^\nu$ on Minkowski spacetime, we shall first show that this extends to a transformation at timelike infinity (for definiteness, the same proof would apply at spatial infinity). For this, we shall consider the expression in the coordinates defined by Eq.~\eqref{cuts: beig-schmidt chart}, $X^\mu = \rho^{-1}\tilde{P}^{\mu}(\y)$, let us write $(\rho', \bm{y}')$ the coordinates of the image point, $(X^\mu)' = (\rho')^{-1}\tilde{P}^{\mu}(\y')$ , then for $\rho$ sufficiently small:
\[\begin{cases} \rho'=  \rho \left(1-2\rho T\cdot m\tilde{P}(\bm{y})- \rho^2T\cdot T\right)^{-\frac{1}{2}} 
\\[0.75em] \tilde{P}^{\mu}(\bm{y}')=\rho'\left(T^{\mu}+\rho^{-1} m^{\mu}{}_{\nu} \tilde{P}^{\nu}(\bm{y}) \right)
\end{cases}. \]
In the limit $\rho \to 0$, one sees that $\tilde{P}^{\mu}(\bm{y}') = m^\mu{}_\nu\tilde{P}^\nu(\bm{y}) + O(\rho)$ and $\rho' =  \rho + \rho^2  \,T \cdot \tilde{P}(\bm{y}') + O(\rho^3)$. Therefore, the Poincaré transformation restricts to a Lorentz transformation on the hyperboloid at infinity. It thus acts as an asymptotic symmetry and we can therefore consider the transformation that it induces on $\Ti$ by the construction in Proposition~\ref{Proposition: Spi group as automorphism}. Using Eq.~\eqref{eq:InducedSymmetry}, one finds:
\[(u,\bm{y}) \mapsto (u - T \cdot \tilde{P}(\bm{y}'), \bm{y}').\]
Observe that if $m^\mu{}_\nu = \delta^\mu_\nu$ and we consider a cut $\xi_X$ of $\Ti$ then this is consistent with:
\[\xi_X \mapsto  \xi_{X+T}.\]
\end{proof}

\section{Carrollian geometry of Spi/Ti and the BMS group}\label{Section: Carrollian geometry of Spi}

At null infinity, it is well known \cite{ashtekar_radiative_1981,ashtekar_geometry_2015} that radiative data can be encoded in terms of (equivalence class of) Carrollian connections, with flat connections being identified with gravity vacua; see \cite{Herfray:2021qmp} for more on the relation between null infinity and Carroll geometry. The situation at spatial/time infinity might seem quite different, we will however show that they are in close analogy: $\Ti$ and $\Spi$ are always strongly Carrollian, i.e. equipped with a compatible torsion free connection. What is more, the appearance of the BMS group, as a subgroup of the SPI group, stems from this extra geometric structure.

As was first shown by Troessaert in \cite{troessaert_bms4_2018} for the linearised theory, the BMS group can be realised as a group of asymptotic symmetries at spatial infinity. Corresponding, non linear, asymptotic conditions (in Hamiltonian formalism) were described in \cite{Henneaux:2018cst,Henneaux:2019yax}. These involve parity conditions on the asymptotic data that generalise the Regge-Teitelboim conditions \cite{Regge:1974zd}. Many of these developments have been inspired and motivated by Strominger's matching condition \cite{strominger_bms_2014} at spatial infinity. These conditions define a global BMS group, acting simultaneously at future and past null infinity; it is in turn the global nature of this group which allows BMS symmetries to constrain scattering observables \cite{He:2014laa}. Building upon these works and motivations, there has been a lot of recent progress on the global nature of conservation laws at spatial and timelike infinity, see \cite{troessaert_bms4_2018,Henneaux:2018gfi,Henneaux:2018mgn,Prabhu:2019fsp,AliMohamed:2021nuc,Prabhu:2021cgk,Chakraborty:2021sbc,Capone:2022gme,Fuentealba:2022xsz,Mohamed:2023jwv,Fuentealba:2023rvf,Mohamed:2023ttb,Ashtekar:2023zul,Fuentealba:2023hzq,Compere:2023qoa,Fuentealba:2023syb,Ashtekar:2024aa,Gasperin:2024bfc,Fuentealba:2024lll,Fiorucci:2024ndw}. In particular, Strominger's matching conditions have been shown to be closely related to parity conditions at spatial infinity, see \cite{Donnay:2023mrd} for a review of these results in relations to S-matrix observables.

As is explained in Appendix \ref{section: appendix Ti/Spi}, it turns out that the parity conditions associated to Strominger's matching are in fact very naturally realised in the model for $\Spi$. In the curved case, as we shall see, these are also neatly realised in terms of a symmetry requirement on the Carrollian connection at $\Spi$. Therefore, not only BMS symmetries are encoded by the Carrollian connection at $\Spi$ but matching conditions stems from preserving, in the curved case, a discrete symmetry of the model.

\subsection{Carrollian connections and their flat models}\label{ssection: Carrollian connections and their flat models}

As we discussed in Section \ref{sssection: Carrollian geometry}, $\Spi/\Ti^{\pm} \to \mathcal{I}^0/\mathcal{I}^{\pm}$ are equipped with a weakly Carrollian geometry $(h_{ab}, n^a)$. We shall see that, in fact, we have more: there is always a canonically induced Carrollian connection on $\Spi$/$\Ti^{\pm}$, i.e. a torsion free connection $\nabla$ satisfying
\begin{align}\label{Carrollian geometry of Spi: Carrollian connection}
 \nabla_c h_{ab}& =0, & \nabla_c n^a&=0.
\end{align}
The data $(h_{ab}, n^a,\nabla)$ constitutes a strongly Carrollian geometry in the sense of \cite{duval_carroll_2014}. As will briefly be reviewed at the end of this subsection, when the connection is suitably flat, the group of residual automorphisms coincides with the Poincaré group $ISO(n,1) = \textrm{SO}(n,1) \ltimes \mathbb{R}^{n+1}$.

We recall (see e.g. \cite{bekaert_connections_2018,Herfray:2021qmp}) that, due to the degeneracy of the metric, connections satisfying \eqref{Carrollian geometry of Spi: Carrollian connection} are not unique. However, if $\nabla$ and $\tilde{\nabla}$ are two Carrollian connections they can only differ by a tensor of the form
\correc{\begin{equation}
   ( \nabla_c -\tilde{\nabla}_c )^a{}_b = C_{bc}\;n^a \label{Eq::Cfirst}
\end{equation}}
where $C_{bc}$ is symmetric and satisfies, $n^b C_{bc}=0$. Let $(u, y^{\alpha})$ be adapted coordinates on $\Spi/\Ti^{\pm}$ i.e. such that
\begin{align}
    n^a &= \partial_u, & h_{ab} &= h_{\alpha\beta}(y) dy^{\alpha} dy^{\beta},
\end{align}
then the Christoffel symbols of any Carrollian connection $\nabla$ on $\Spi$/$\Ti$ must be of the form
\begin{align}\label{Carrollian geometry of Spi: Carrollian connection coefficients}
   (\Gamma_{c})^a{}_b  = \begin{pmatrix}
        \;0\; & -\frac{1}{2}C_{\gamma\beta}\, dy^{\gamma} \\[0.2em]
        \;0\; & \; \Gamma^{\alpha}_{\gamma\beta}\, dy^{\gamma} 
    \end{pmatrix}
\end{align}
where $\Gamma_{\gamma\beta}^{\alpha}$ are the Christoffel symbols of $h_{ab}$ and $C_{\alpha\beta}=C_{(\alpha\beta)}$ is the only free parameter, i.e. unconstrained by \eqref{Carrollian geometry of Spi: Carrollian connection}.

In the following subsection, we will prove that the geometry of $\Spi$/$\Ti$ fixes this ambiguity uniquely.

\paragraph{The flat models}\mbox{} 

We here review the flat models of Carrollian geometry (see e.g. \cite{Herfray:2021qmp}). Let $h_{\alpha\beta}$ be a constant curvature metric i.e. with Riemann tensor\footnote{We can take e.g. $2 \Lambda = \nu_0$ to match with \eqref{Ashtekar Romano: Einstein equations on h}.}
\begin{equation}\label{Base curvature: flat model}
    R^{(h)}{}^{\alpha}{}_{\beta \gamma \delta} = 4\Lambda \, \delta^{\alpha}{}_{[\gamma}h_{\delta]\beta}
\end{equation}
then taking, in the expression \eqref{Carrollian geometry of Spi: Carrollian connection coefficients} for a Carrollian connection,
    \begin{equation}\label{Carrollian connection: flat model}
        -\frac{1}{2}C_{\alpha\beta} = u \,2\Lambda\,h_{\alpha \beta}
    \end{equation}
yields the flat models of Carrollian geometry. These connections can be invariantly characterised as those with constant curvature in the sense that 
\begin{align}\label{Carrollian curvature: flat model}
R^{a}{}_{bcd}=4\Lambda \, \delta^{a}{}_{[c}h_{d]b}.
\end{align}
For $\Lambda$ positive/null/negative these respectively give the Carroll--dS / Carroll / Carroll--AdS models \cite{bacry_possible_1968,figueroa-ofarrill_spatially_2019}. At this stage we have not said anything about the signature of $h_{\alpha\beta}$: in the literature the names Carroll--dS / Carroll / Carroll--AdS tend to be restricted to models with Euclidean signature. In what follows, we will consider the model corresponding to $\Ti$, which is genuinely Carroll--AdS with base space $H^n$, and $\Spi$, which is a pseudo Carroll--AdS space with base space $dS_n$.

Symmetries of the strongly Carrollian geometry are vector fields $X$ satisfying\footnote{We remind the reader that the Lie derivative of a connection $\left(\mathcal{L}_X \nabla_c\right)^a{}_b$ is a tensor and, when the connection is torsion-free, is symmetric in $b$ and $c$.}
\begin{align*}
    \mathcal{L}_X h_{ab}&=0, & \mathcal{L}_X n^a&=0,  & \left(\mathcal{L}_X \nabla_c\right)^a{}_b &= 0.
\end{align*}
Making use of the identity for the Lie derivative of a torsion-free connection
\begin{equation}\label{Carrollian connection: Lie derivative of connection}
    \left(\mathcal{L}_X \nabla_c\right)^a{}_b = \nabla_c\nabla_b X^a -  R^a_{\phantom{a}bcd} X^d,
\end{equation}
 the third condition is equivalent to $\nabla_c\nabla_b X^a =2 \Lambda\left( \delta^a{}_c X^d  h_{db} - X^a h_{cb}\right)$. From these conditions we find that a symmetry is of the form
\begin{equation*}
    X =  \omega(\y) \partial_u + \chi^{\alpha}(\y) \partial_{\alpha} ,
\end{equation*}
where $\chi^{\alpha}(\y)$ and $\omega(\y)$ must satisfy:
\begin{align}
   D_{(\alpha} \chi_{\beta)}&=0, & \left(D_{\alpha}D_{\beta} + 2\Lambda h_{\alpha\beta} \right) \omega &=0.
\end{align}
These can be solved explicitly as follows, when $\Lambda \neq 0$. Let $\y^{\alpha} \to P^{\mu}(\y)$ with $P(\y)^2 = (2\Lambda)^{-1}$ be an embedding of $H^n$/$dS_{n}$ into $\mathbb{R}^{n,1}$ (e.g. as in \eqref{Massive scalar field: P tilde definition}). The Carroll symmetries are given by
\begin{align}
    \chi_{\alpha}(\y) &= \partial_{\alpha} P^{\mu}(\y) M_{[\mu\nu]} P^{\nu}(\y), &\omega(\y) = P(\y)^{\mu} T_{\mu}
\end{align}
where $(M_{[\mu\nu]},T^{\mu}) \in \mathfrak{so}(n,1) \times \mathbb{R}^{n,1}$. One can check that the resulting algebra of vector fields coincides with the Poincaré algebra
\begin{equation}
    \mathfrak{iso}(n,1) = \mathfrak{so}(n,1) \ltimes \mathbb{R}^{n,1}.
\end{equation}
In order to check these results, it is useful to have in mind the identities\footnote{This follows from the covariant derivative $D_\alpha$ of a one-form on the hyperboloïd, induced from the $n+1$ dimensional metric. Explicitly, for vectors, $D_{\alpha}(V^{\beta} \partial_{\beta} P^{\mu}) := \partial_{\alpha}\left(V^{\beta} \partial_{\beta} P^{\mu} \right)_{\perp} = \partial_{\alpha}\left(V^{\beta} \partial_{\beta}  P^{\mu} \right) + 2\Lambda V^{\alpha}h_{\alpha \beta } P^{\mu}$, where $\perp$ indicates projection on the tangent space of the hyperboloïd.} satisfied by the \emph{coordinate functions} of the embedding given by $P^{\mu}$:
\begin{align}\label{Flat model: identity on P}
    D_{\alpha}D_{\beta} P(\y)^{\mu} &= - 2 \Lambda h_{\alpha\beta} P(\y)^{\mu},&
    h^{\alpha \beta}D_\alpha  P(\y)^\mu  D_\beta P(\y)^\nu &= \eta^{\mu \nu} - 2 \Lambda P(\y)^\mu P(\y)^\nu.
\end{align}

\subsection{Carrollian geometry of Spi/Ti}\label{ssec:carroll}

We will realise the Carrollian connection in terms of the asymptotic data \eqref{Ashtekar-Romano: metric Coordinates expansion}. In order to do this, we need to relate coordinates on $\Spi/\Ti$ with Beig-Schmidt charts.

\begin{proposition}
    A coordinate system $(u, {\y}^{\alpha})$ on $\Spi/\Ti$ naturally defines a Beig-Schmidt chart\footnote{At least up to order $O(\rho^3)$ in $\rho$. It is in fact well known from \cite{beig_einsteins_1982} that the leading orders $\rho$ of a boundary defining function defines a unique Beig-Schmidt chart to every order, but we will not need this result here.} $(\rho, {y}^{\alpha})$. If $(u_1, {\y_1}^{\alpha})$ and $(u_2, {\y_2}^{\alpha})$ are two such coordinate systems with $u_2 = u_1 + \omega$ then 
    \begin{align}\label{Carroll connection: action of supertranslation on carroll coordinates}
        u_2 &= u_1 + \omega & \Leftrightarrow && \rho_2 &= \rho_1 - \omega \rho_1^2 + O(\rho_1^3)
    \end{align}
    and the corresponding coefficients $k_{\alpha\beta}$ and $\sigma$, defined in these charts by \eqref{Ashtekar-Romano: metric Coordinates expansion}, are related via
\begin{align}\label{Carrollian connection: k transformation rules}
    k_2{}_{\alpha\beta}(\y_2) \;d\y_2^{\alpha} d\y_2^{\beta}&= \Big( k_1{}_{\alpha\beta} - 2\left( h_{\alpha \beta}\omega + \nu_0^{-1}D_{\alpha}D_{\beta}\omega\right)  \Big)(\y_1) \; d\y_1^{\alpha} d\y_1^{\beta},&
    \sigma_1(\y_1) & = \sigma_2(\y_2).
\end{align}
\end{proposition}
\begin{proof}\mbox{}

    The set $\{u=0\}$ for the chosen coordinate system defines a section of $\Spi/\Ti$. By definition, this section defines, in turn, a boundary defining function $\rho$ up to order 3: $\rho  \mapsto \rho  + O(\rho^3)$. 

    Let $u_1$ and $u_2 = u_1 +\omega$ be two charts. The corresponding sections are given, when written in the first chart, respectively by the equations $u_1=0$ and $u_1 = -\omega$. One concludes that the second is obtained from the first via the $\mathbb{R}$-action $u\mapsto u - \omega$. Once more by definition, this $\mathbb{R}$-action takes the corresponding boundary function $\rho_1$ to $\rho_1 - \omega \rho^2$ and thus $\rho_2 = \rho_1 - \omega \rho_1^2 + O(\rho_1^3)$. It is then always possible to choose this $O(\rho_1^3)$ such that the corresponding chart satisfies $g_{\rho\rho} = \rho^{-2}\nu_0^{-1}(1+ \rho \sigma)^2$ to order three in $\rho$. The chart $(\rho,y^{\alpha})$ itself is then obtained by extending $y^{\alpha}$ from the boundary to the spacetime through the requirement $N^a\nabla_a y^{\alpha} =0$. Here $N^a$ is the vector field appearing in definition \ref{Definition:AsymptoticFlatness}.

   One can now directly compute $k_1$, $k_2$, $\sigma_1$ and $\sigma_2$ in these charts to find the results or adapt \eqref{Ashtekar Romano: transformation of k}, the sign discrepancy corresponding to a difference in passive/active action of the diffeomorphism.
    
\end{proof}

\subsubsection{Carrollian connection}

\begin{proposition}\label{Proposition: Carrollian connection}
    Let $(u, \y^{\alpha})$ be a coordinate system on $\Spi/\Ti$.  Let $( \rho, \y^{\alpha})$ be the associated Beig-Schmidt chart and $k_{\alpha\beta}(\y)$ be the object appearing in the corresponding expansion \eqref{Ashtekar-Romano: metric Coordinates expansion}.
    
    Taking    \begin{align}\label{Carrollian connection: expression in terms of k!}
        -\frac{1}{2}C_{\alpha\beta} &= \nu_0\left(\,  \frac{1}{2} k_{\alpha\beta}\, + u\;  h_{\alpha \beta}\right),
    \end{align}
    in the expression \eqref{Carrollian geometry of Spi: Carrollian connection coefficients} of a Carrollian connection then defines a Carrollian connection on $\Spi/Ti$. Moreover, the result does not depend on the coordinates chosen.
\end{proposition}
\begin{proof} \mbox{} 

The coordinate expression \eqref{Carrollian connection: expression in terms of k!} defines, together with \eqref{Carrollian geometry of Spi: Carrollian connection coefficients}, a Carrollian connection. Let us now show that it does not depend on the coordinate system. If $(u, {\y}^{\alpha})$ and $(u', {\y'}^{\alpha})$  are related via $u' = u+\omega$ we have, making use of \eqref{Carrollian connection: k transformation rules}, that upon such change of coordinates,
\begin{align*}
    \left(\frac{1}{2}k_{\alpha \beta}\right)' &= \frac{1}{2}k_{\alpha \beta}  - ( \,h_{\alpha \beta} + \nu_0^{-1} D_{\alpha}D_{\beta})\omega, & u' h_{\alpha \beta} &= u h_{\alpha \beta} + \omega h_{\alpha \beta},
\end{align*}
and thus, if we define $C_{\alpha \beta}$ via \eqref{Carrollian connection: expression in terms of k!},
\begin{align*}
   (-\frac{1}{2}C_{\alpha \beta})' \;&\;:= \nu_0\left( (\frac{1}{2}k_{\alpha \beta})' + u'  h_{\alpha \beta}\right) \; =\; -\frac{1}{2}C_{\alpha \beta}  \,-\,D_{\alpha}D_{\beta}\omega.
\end{align*}
This is the consistent transformation law for a Carrollian connection under the change of coordinate $u'= u+\omega$. We therefore obtain that \eqref{Carrollian connection: expression in terms of k!} transforms consistently under change of coordinates and thus that the resulting Carrollian connection is independent of this choice.
\end{proof}

\begin{remark}
If we apply the formula\footnote{After performing the appropriate rescaling and coordinate change to match the conventions of this paper.} in the fibre for the connection coefficient $\upsilon_{\alpha\beta}$ constructed in~\cite{Borthwick:2024wfn}, wherein the connection was obtained from geometric considerations in the absence of mass aspect, we would arrive at:
\[\begin{aligned} -\frac{1}{2}C_{\alpha\beta}&=-\lim_{\rho \to 0}\nu_0\rho^{-1}\left(\nu^{-1}\frac{1}{\rho}\Gamma^\rho_{\alpha\beta} - (h_{\alpha\beta} + \rho(k_{\alpha\beta}-2\sigma h_{\alpha\beta})+o(\rho)\right)+ u\nu_0h_{\alpha\beta}\\&= \nu_0\left(\frac{1}{2}k_{\alpha\beta} + (u-\sigma) h_{\alpha\beta}\right).\end{aligned}\]
This coincides with~\eqref{Carrollian connection: expression in terms of k!} up to a shift of the trace of $k$ by $\sigma$. Note that this possible alternative definition for the Carrollian connection suggests that the gauge fixing $k-2n \sigma =0$ might be geometrically more natural, in the presence of mass, than the choice $k=0$ which has been used extensively in the literature.
\end{remark}

\subsubsection{Carrollian Curvature}
It is a straightforward computation that the curvature $R^{a}{}_{bcd}$ of the Carrollian connection at $\Spi/\Ti$ is given, in a coordinate system $(u,\y^{\alpha})$ by:
\begin{equation*}\label{Carrollian connection: Carrollian curvature}
R^a_{\,\,b}=\begin{pmatrix}\;0\; & \frac{\nu_0}{2} D_\delta k_{\beta\gamma}\,dy^\delta\wedge dy^\gamma + \nu_0h_{\beta\gamma}\, du\wedge dy^\gamma \\[0.2em] \;0\; & {R^{(h)}}^\alpha_{\,\,\beta}  \end{pmatrix}.
\end{equation*}

Note that the difference $R^a{}_{bcd} - 2 \nu_0\, \delta^{a}{}_{[c}h_{d]b}$ between this curvature and the curvature of the model \eqref{Carrollian curvature: flat model} is proportional to
\begin{align}\label{Carrollian curvature: Cartan curvature}
R^{(h)}{}^{\alpha}{}_{\beta \gamma \delta} - 2 \nu_0\, \delta^{\alpha}{}_{[\gamma}h_{\delta]\beta}, && \textrm{and}&& D_{[\gamma} k_{\delta]\beta}. 
\end{align}
In the situation at hand, these in fact parametrise the curvature of a Cartan connection modelled on \eqref{Ti and Spi: flat models}, see \cite{Herfray:2021qmp}; in particular, the vanishing of the above two tensors is the necessary and sufficient condition for $\Spi/\Ti$ to be locally isomorphic to the models (and, for $n=3$, imply the leading Einstein's equations \eqref{Ashtekar Romano: Einstein equations on h}, \eqref{Ashtekar Romano: Einstein equations on k} for the spacetime; in higher dimensions, the left hand side of \eqref{Ashtekar Romano: Einstein equations on k} then vanishes and yields a further constraint on $\sigma$). 

The curl of $k$ appearing here is also well known to be related to the asymptotic Weyl tensor of the spacetime, see e.g. \cite{Compere:2023qoa}. In most of the literature, including \cite{Compere:2023qoa}, it is assumed that the above curvature vanishes. Therefore, the induced Carrollian geometry must be flat and there exists $\mathcal{C}(\y) \in \mathcal{C}^{\infty}(dS_n/H^n)$ such that
\begin{equation}\label{Carrollian connection: pure gauge condition}
    k_{\alpha\beta} =  2\left( \nu_0^{-1}D_{\alpha} D_{\beta} + h_{\alpha \beta} \right) \mathcal{C}.
\end{equation}

\subsection{BMS group at Ti/Spi}\label{ssection: BMS group at Ti/Spi}

For definiteness we assume, in this section, that we are in the most physically relevant situation where $\mathcal{I}^{0}\simeq dS_n$ / $\mathcal{I}^{\pm}\simeq H^n$.

\subsubsection{Carrollian symmetries}

As was recalled in Proposition \ref{Proposition: Spi group as automorphism}, the $\SPI$ group coincides with automorphisms of the weak Carrollian geometry $(h_{ab}, n^a)$. Preserving the (weak) Carrollian geometry
\begin{align*}
    \mathcal{L}_X h_{ab}&=0, & \mathcal{L}_X n^a&=0, 
\end{align*}
is equivalent to 
\begin{equation*}
    X^a =  \omega(\y) \partial_u + \chi^{\alpha}(\y) \partial_{\alpha},
\end{equation*}
where $\chi^{\alpha}(\y) \partial_{\alpha}$ is a Killing vector for the metric, $\mathcal{L}_{\chi} h_{\alpha\beta }=0$.

In general, strongly Carrollian geometries $(h_{ab}, n^a ,\nabla_b)$ do not have any automorphisms: the condition 
\begin{equation*}
    \mathcal{L}_X \nabla_c = 0,
\end{equation*}
is generically too strong because the connection does not have to be flat. However, the Carrollian connections \eqref{Carrollian connection: expression in terms of k!} induced on $\Ti/\Spi$ are rather special so we might be missing something. Let us make use of the identity \eqref{Carrollian connection: Lie derivative of connection} to evaluate the corresponding Lie derivative $(\mathcal{L}_X \nabla_c)^a{}_b$ explicitly
\begin{equation}\label{Carrollian connection: generic transformation of the connection}
(\mathcal{L}_X \nabla_c)^a{}_{b}=\begin{pmatrix}\;0\; & (D_\gamma D_\beta \omega +\nu_0 h_{\beta\gamma}\omega + \frac{\nu_0}{2}\mathcal{L}_\chi k_{\beta\gamma})\, \nabla_c y^\gamma \\[0.2em] \;0\; & (\mathcal{L}_\chi D_\gamma)^\alpha{}_\beta \,\nabla_c y^\gamma \end{pmatrix}.
\end{equation}
Since it is assumed that $\chi^{\alpha}$ is a Killing vector for $h_{\alpha\beta}$, it also is a symmetry of its Levi-Civita connection $D$ and the second line above is automatically vanishing. Therefore, one sees that $X^a$ is a symmetry if and only if

\begin{equation}\label{Carrollian connection: full automorphisms}
 \big(D_\beta D_{\gamma} + \nu_0 h_{\beta\gamma}\Big)\,\omega  = -\frac{\nu_0}{2}  \mathcal{L}_\chi k_{\beta\gamma}.
\end{equation}
If this equation admits any solution then we in fact obtain a four parameter family of solutions by further solving the homogeneous equation $\big(D_\beta D_{\gamma} + \nu_0 h_{\beta\gamma}\Big)\,\omega  = 0$. This would imply that the geometry is maximally symmetric and therefore coincide with the flat model \eqref{Ti and Spi: flat models}. In general, however, the geometry is not flat and this equation has no solution. 

\subsubsection{BMS symmetries at Ti and representation on massive fields}

As just discussed, in general, there will be no symmetries of the strongly Carrollian geometry $(h_{ab}, n^a ,\nabla_b)$. 

Nevertheless, the Carrollian connections that we are considering are, once again, rather special: from \eqref{Carrollian connection: generic transformation of the connection} their Lie derivatives $(\mathcal{L}_X\nabla_c)^a{}_b$ always satisfy
\begin{equation*}
    n^c(\mathcal{L}_X\nabla_c)^a{}_b =0.
\end{equation*}

Therefore, one sees that we can always consider the weaker symmetry condition, inspired from Compère--Dehouck \cite{Compere:2011ve} and Troessaert \cite{troessaert_bms4_2018},
 \begin{align}\label{Carrollian Connection: BMS symmetry condition}
h^{cb}(\mathcal{L}_X\nabla_c)^a{}_b &=0,
\end{align}
which, from  \eqref{Carrollian connection: generic transformation of the connection}, is equivalent to
\begin{equation}\label{Carrollian connection: weaker automorphisms}
( D^2 + \nu_0 n) \omega(\y) = -\frac{\nu_0}{2} \mathcal{L}_{\chi} k(\y).
\end{equation}
where $k:= h^{\alpha\beta}k_{\alpha \beta}$. Note that, when expressed in coordinates $(u,\y^{\alpha})$ on $\Ti$, the condition \eqref{Carrollian Connection: BMS symmetry condition} simply means that the trace $k:= h^{\alpha\beta}k_{\alpha \beta}$ is preserved under such transformations and one thus recovers the condition from \cite{Compere:2011ve,troessaert_bms4_2018}.

Contrary to \eqref{Carrollian connection: full automorphisms}, equation \eqref{Carrollian connection: weaker automorphisms} always admits solutions. These have already been worked out several times in the literature, see e.g. \cite{Compere:2023qoa}, we only review their essential features.

At time-like infinity, $\nu_0=-|\nu_0|$, any regular solution on the hyperboloïd $H^n$ must be of the form
\begin{equation*}
    \omega(\y) = \varpi(\y) + \omega_\xi(\y),
\end{equation*}
where $\varpi(\y)$ is any particular solution while $\omega_{\xi}(\y)$ is entirely parametrised by a free function $\xi(z^i)$ on the sphere $S^{n-1} = \partial(H^n)$ and given by the bulk-to-boundary propagator \eqref{Asympotic symmetries: BMS in SPI}.  The algebra of symmetries obtained in this way then coincides with the BMS group. This therefore proves the following.

\begin{proposition}\label{Proposition: BMS group inside SPI}
    The BMS group coincides with the subgroup of diffeomorphisms of $\Ti$ leaving $(h_{ab},n^a, [\nabla_a])$ invariant, where $[\nabla_a]$ is the equivalence class of connection given by the induced Carrollian connection $\nabla_a$ on $\Ti$ and the equivalence relation
    \begin{equation*}
        \nabla_a \sim \nabla_a  + T_{bc} n^a, \qquad\textrm{where}\quad n^bT_{cb}=0,\quad h^{cb}T_{cb}=0.
    \end{equation*}
\end{proposition}

By means of a $\SPI$ supertranslation, one can in fact always choose a coordinate system $(u,\y^{\alpha})$ on Ti such that $k=0$. These coordinates then realise an explicit isomorphism $BMS_{n+1} \simeq  SO(n,1)  \ltimes \mathcal{E}[1](S^{n-1})$ given, at the level of the algebra, by
\begin{equation}\label{Carrollian connection:  BMS at Ti}
   (\,\xi , \chi^{\alpha}(\y)\,)  \quad \mapsto \quad X^a =  \omega_{\xi}(\y) \partial_u + \chi^{\alpha}(\y) \partial_{\alpha} 
\end{equation}
with $\chi^{\alpha}(\y) \partial_{\alpha} \in \mathfrak{so}(n,1)$ a Killing vector on $H^n$ and $\xi(z^i) \in \mathcal{E}[1](S^{n-1})$ a weight one conformal density on the celestial sphere, and $\omega_{\xi}(\y)$ given by \eqref{Asympotic symmetries: BMS in SPI}.

\paragraph{Longhi--Materassi BMS representation on massive fields}\mbox{}

Combining the previous discussion with Proposition \ref{Proposition: Action of the SPI group on massive fields}, we obtain a geometrical representation of the BMS group on massive fields.
\begin{proposition}\label{Proposition: Action of the BMS group on massive fields}
    Massive fields at Ti form a (unitary irreducible) representation of the BMS group realised by
    \begin{align*}
      \varphi\big(u, \y\big) &\;\mapsto   \varphi\big(u +\omega_{\xi}(\y), \y'(\y)\big) && \Longleftrightarrow &  a(\y) &\;\mapsto a(\y') e^{-im \omega_{\xi}(\y)},
   \end{align*}
   where $\omega_{\xi}(\y)$ is the function on $H^n$ given by \eqref{Asympotic symmetries: BMS in SPI}.
\end{proposition}
This representation is isomorphic to the canonical realisation of Longhi--Materassi \cite{Longhi:1997zt}. \correc{As we previously pointed out, from the perspective of McCarthy's classification of BMS UIRs~\cite{Mccarthy:1972ry,McCarthy_72-I,McCarthy_73-II,McCarthy_73-III,McCarthy_76-IV,McCarthy_75,McCarthy_78,McCarthy_78errata,Girardello:1974sq}), this representation is rather specific: a precise description of how this representation fits into the classification can be found in \cite{Bekaert:2024jxs,Bekaert:2025kjb}, see in particular section 5 of \cite{Bekaert:2025kjb}. For completeness, we summarise here some relevant points: first, as should be clear from the present exposition for example, the Longhi--Materassi BMS UIR is obtained by extending the usual (scalar) massive representation of the Poincaré group to the BMS group. In fact, it follows from McCarthy's work that all UIRs of the Poincaré group similarly lift to BMS UIRs: these are the hard BMS representations of \cite{Bekaert:2024jxs,Bekaert:2025kjb} and, by definition, they are therefore in 1-1 correspondence with usual (Poincaré) particles. Generic BMS particles (i.e. BMS UIR) have a much less straightforward interpretation as a quantum superposition of usual particles in all possible gravity vacua, see \cite{Bekaert:2024jxs}. Secondly, hard representations are almost uniquely singled out by the requirement that their BMS little group and their Poincaré little group coincide (see Theorem 5.1 in \cite{Bekaert:2025kjb}). For the massive hard representation of Proposition \ref{Proposition: Action of the BMS group on massive fields} this is in fact exactly the case: the Longhi--Materassi representation of Proposition \ref{Proposition: Action of the BMS group on massive fields} is the unique BMS UIR with BMS little group $SU(2)$ (in the present article we restricted ourselves to scalar fields but this results extends straightforwardly to all spins, with BMS spins corresponding to a choice of UIR of the little group).}

\subsubsection{BMS symmetries at Spi and the matching conditions}

We recall that the metric defining de Sitter space $dS_n$ with scalar curvature $n (n-1) \nu_0$ is given, in the standard coordinate chart $\y^{\alpha}= (\psi,\vartheta)\in \mathbb{R}\times S^{n-1}$, by:
\[h_{dS_n}=\frac{1}{\nu_0}\left( -d\psi^2 + \cosh^2\psi \,d\Omega^{n-1}\right).\]
In particular, it comes with a parity map $\Upsilon:  dS_n \to dS_n$, unique up to $SO(n,1)$, given by
\begin{equation}
 \Upsilon:   = (\psi , \vartheta) \mapsto (-\psi , a^*\vartheta),
\end{equation}
where $a: S^{n-1} \to S^{n-1}$ is the antipodal map on the celestial sphere.

At space-like infinity, $\nu_0=|\nu_0|$ and regular solutions of \eqref{Carrollian connection: weaker automorphisms} on de Sitter $dS^n$ must be of the form (see e.g. \cite{Henneaux:2018mgn,Compere:2023qoa})
\begin{equation}\label{Carrollian connection: BMS at Spi}
    \omega(\y) = \varpi(\y) + \omega^O_{\xi}(\y) +\omega^E_{\xi'}(\y),
\end{equation}
where $\varpi(\y)$ is any particular solution, $\omega^O_{\xi}(\y)$ and $\omega^E_{\xi'}(\y)$ are respectively odd and even for the parity
\begin{align*}
    \Upsilon^* \omega^O_{\xi} &= - \omega^O_{\xi}, & \Upsilon^* \omega^E_{\xi'} &= \omega^E_{\xi'},
\end{align*}
and are parametrised by free functions $\xi(z)$ and $\xi'(z)$ on the sphere $S^{n-1}$. In other words, while at $\Ti$ the condition \eqref{Carrollian Connection: BMS symmetry condition} was enough\footnote{Together with the crucial requirement that the solution be globally defined on time-like infinity $\mathcal{I}^{\pm} \simeq H^n$.} to single out the BMS group as a group of symmetry, at $\Spi$ this will generically lead to twice as many free parameters. In order to select a unique BMS group, one in fact needs, as first explained by Troessaert \cite{troessaert_bms4_2018}, to impose an extra parity condition. We will here interpret this condition in terms of the Carrollian geometry of $\Spi$.\\

As explained in appendix \ref{section: appendix Ti/Spi}, there is a natural discrete symmetry $P$ acting on the model of $\Spi$. In particular, this implies that the (strong) Carrollian geometry of the model is invariant under this discrete symmetry.  

Having in mind the goal of generalising this property in the curved setting, we introduce a notion of parity for Carrollian geometry.
\begin{definition}
    We will say that a map $P : \Spi \to \Spi$ is a parity of $\Spi$ if \begin{enumerate}[i)]\item it is an involution, $P^2 = id$, \item it satisfies
\begin{align*}
    \Big( P_* n^a, P^* h_{ab} \Big) =  \Big( -n^a, h_{ab} \Big),
\end{align*}
\item the map it induces on $dS^n$ is a parity i.e. a non-trivial, time reversing, orientation preserving involution of $dS^n$ for $n$ odd (orientation reversing for $n$ even). \end{enumerate} 
\end{definition}
If $(u,\y^{\alpha})$ is a set of coordinates on $\Spi$ then a parity must be of the form
\begin{equation}\label{Carrollian connection: parity on Spi}
 P:   \Big( u, \y^{\alpha}\Big) \mapsto \Big( \mathcal{C}(\y) -u, \Upsilon^* \y^{\alpha}\Big)
\end{equation}
where $\mathcal{C}$ is even under de Sitter  parity $\Upsilon^* \mathcal{C} = \mathcal{C}$. This last condition follows from the requirement that $P$ is an involution. Indeed, we have
\begin{equation}
 P^2:   \Big( u, \y^{\alpha}\Big) \mapsto \Big( u - \mathcal{C}(\y) + \Upsilon^*\mathcal{C}(\y) , \y^{\alpha}\Big).
\end{equation}

\begin{definition}[Even Carrollian geometry at Spi]\label{Definition: Odd Carrollian geometry}\mbox{}

A (strongly) Carrollian geometry $(n^a, h_{ab}, \nabla_a)$ at $\Spi$ will be said to be even if there exists a parity $P : \Spi \to \Spi$ such that
\begin{equation}\label{Carollian connection: parity condition, covariant}
    P^* \nabla = \nabla.
\end{equation}   
\end{definition}
As discussed in appendix \ref{section: appendix Ti/Spi}, globally flat Carrollian geometries on $dS_{n}$ are always even. In general however, a Carrollian geometry $(n^a, h_{ab}, \nabla_a)$ need not satisfy this property: the requirement that such a symmetry exists is an extra condition restricting the space of admissible connections. In the coordinate system $(u,\y^{\alpha})$ such that $P$ is of the form \eqref{Carrollian connection: parity on Spi} the condition \eqref{Carollian connection: parity condition, covariant} reads
\begin{equation}
    \Upsilon^*\Big( k_{\alpha\beta} + 2(\nu_0^{-1}D_{\alpha}D_{\beta} + h_{\alpha\beta})\mathcal{C}^E(\y) \Big) = - \Big(k_{\alpha\beta}+ 2(\nu_0^{-1}D_{\alpha}D_{\beta} + h_{\alpha\beta})\mathcal{C}^E(\y)\Big), 
\end{equation} 
where we here emphasised that $\mathcal{C}^E = \frac{1}{2} \mathcal{C}$ must be even. In other terms, $k_{\alpha\beta}(\y)$ must be odd under parity of $dS_{n}$, possibly up to a pure gauge even contribution. This is essentially the parity condition at spatial infinity from \cite{Compere:2023qoa} and is very reminiscent of the boundary conditions imposed on the Cauchy data at spatial infinity by Henneaux and Troessaert \cite{Henneaux:2018cst,Henneaux:2019yax}, see also \cite{Mohamed:2021rfg,Mohamed:2023jwv,Mohamed:2023ttb}. On can always adapt our coordinate system to choose $2 \mathcal{C}^E =\mathcal{C}=0$  and this recovers the parity condition from \cite{Compere:2023qoa}
\begin{equation}
    \Upsilon^* k_{\alpha\beta}= -k_{\alpha\beta}.
\end{equation} 
Note, however, that we never had to require here that the the connection is flat, which is always implied in \cite{Compere:2023qoa}.

By construction, in this preferred set of coordinate, the subset of \eqref{Carrollian connection: BMS at Spi} preserving the extra symmetry will be of the form $\omega(\y) =  \varpi(\y)+ \omega^O_{\xi}(\y)$.  Similarly to the situation at $\Ti$, we can now improve these coordinates to set $k=0$ and this then realises an isomorphism $BMS_{n+1}\simeq SO(n,1) \ltimes  \mathcal{E}[1](S^{n-1}) $:
\begin{equation*}
   (\,\xi , \chi^{\alpha}(\y)\,)  \quad \mapsto \quad X^a =  \omega_{\xi}^O(\y) \partial_u + \chi^{\alpha}(\y) \partial_{\alpha}.
\end{equation*}

We thus proved
\begin{proposition}\label{Proposition: BMS group inside SPI (2)}
Let $(n^a, h_{ab}, \nabla_a)$ be an even Carrollian geometry at $\Spi$. The BMS group coincides with the subgroup of diffeomorphisms of $\Spi$ leaving invariant $(h_{ab},n^a, [\nabla_a])$ where $[\nabla_a]$ is the equivalence class of \emph{even} connections given by the induced Carrollian connection $\nabla_a$ on $\Spi$ and the equivalence relation
    \begin{equation*}
        \nabla_c \sim \nabla_c  + T_{bc} n^a, \qquad\textrm{where}\quad n^bT_{cb}=0,\quad h^{cb}T_{cb}=0 ,\quad \Upsilon^*T_{cb}=-T_{cb}.
    \end{equation*}
\end{proposition}

At this stage, one might want to see how this parity reasoning applies to charges. Similarly, one might wonder how Strominger's matching conditions, that are a consequence of the parity conditions \cite{Henneaux:2018mgn,Mohamed:2021rfg,Capone:2022gme,Mohamed:2023jwv,Mohamed:2023ttb,Compere:2023qoa}, are realised in this geometrical setting. We will discuss these aspects in a subsequent work.

\section{Discussion}

The present work intended to highlight new definitions, Definition \ref{Definition: Spi} and \ref{Definition: Ti}, for extended boundaries $\Spi \simeq \mathbb{R} \times \mathcal{I}^0$ and $\Ti^{\pm} \simeq \mathbb{R} \times \mathcal{I}^\pm$ at spatial and timelike infinity respectively; these were first introduced in \cite{RSTA20230042,Borthwick:2024wfn} for projectively compact spacetimes. Although projective compactness implies that the mass aspect is vanishing (see Section \ref{ssection:Projective compactness}), the present work demonstrates that this assumption can be lifted without any harm. 

Our definition cleanly ties up and relates two similar constructions which have appeared in the literature. As we explained in Section \ref{ssection: Comparison with the previous literature}, when both construction make sense, the present definition at spatial infinity is naturally identified with the extended boundary of Ashtekar-Hansen \cite{hansen_r._o._metric_1978}. What is more, when applied to Minkowski space, the obtained extended boundaries are canonically identified with the homogenous space \eqref{Intro: flat models}
introduced by Figueroa O'Farrill, Have, Prohazka and Salzer.

We also sought to demonstrate how efficient this construction is in capturing the essence of the physics occurring at time-like infinity: the extended boundaries are naturally equipped with  (weak) Carrollian geometries and asymptotic symmetries of the spacetime are identified, through this construction, with Carroll automorphisms, see Proposition \ref{Proposition: Spi group as automorphism}. Moreover, massive fields naturally induce on $\Ti$ an invariant notion of scattering data, see Proposition \ref{Proposition::MassiveScalar} and \ref{Proposition: scattering data are sections}. It follows that scattering data for massive fields on $\Ti$ carry a (unitary irreducible) representation of the asymptotic symmetry group, see Proposition \ref{Proposition: Action of the BMS group on massive fields}. The scattering data are identified with the Carrollian fields which were constructed in \cite{Have:2024dff} by means of representation theory. The present construction allows to go one step further for Minkowski space: spatial geodesics emanating from points of Minkowski space naturally define a four-parameter space of cut of $\Ti$, see equation \eqref{Cut of Ti: image of the cut}. Massive fields are then recovered from the scattering data by a Kirchhoff-type formula involving integration over these cuts, see Proposition \ref{Proposition: Reconstruction formula}.

Finally we proved, see Proposition \ref{Proposition: Carrollian connection}, that $\Ti$ and $\Spi$ are examples of (strongly) Carrollian manifolds; the asymptotic geometry of the spacetime naturally induces on $\Ti$/$\Spi$ a Carrollian connection. This is in complete parallel with the geometry of null infinity \cite{ashtekar_radiative_1981,ashtekar_null_2018,Herfray:2021qmp} and ties up with the known physics at time/spatial infinity in the following way: when the Carrollian geometry is flat, it selects a preferred Poincaré group of symmetry at Time/Spatial infinity (see section \ref{ssection: Carrollian connections and their flat models}). When the geometry is not flat, a weaker notion of symmetry of the Carrollian connection allows to select a preferred BMS group inside the SPI group, see Proposition \ref{Proposition: BMS group inside SPI} and \ref{Proposition: BMS group inside SPI (2)}; this gives a geometrical realisation of Troessaert's condition \cite{troessaert_bms4_2018}. In particular, the parity condition, usually imposed in this context, follows from requiring that the discrete symmetry of the model, described in details in appendix \ref{section: appendix Ti/Spi}, is still present in the curved setting.

We hope that the present article convinced the reader that, despite the seemingly abstract form of Definition \ref{Definition: Spi} and \ref{Definition: Ti}, $\Spi$ and $\Ti$ are very practical constructions in nature and allow to put in a coherent framework physical aspects of spatial and time infinity (massive fields, Carroll geometries, asymptotic symmetries and parity conditions) which might otherwise appear falsely disconnected. 

There is in fact more to this construction. It turns out that equations \eqref{Ashtekar Romano: Einstein equations on k} have a projective invariance which allows us to consistently add a projective boundary $\partial\Ti = \mathbb{R} \times S^{n-1}$ (for $\Spi$ the boundary is disconnected with past and future components) to $\Ti$/$\Spi$ . For Minkowski space these boundaries, are identified with null infinity and, in the curved case, encode gravity vacua. This will be discussed in a subsequent work.

\appendix

\section{Models of Ti and Spi} \label{section: appendix Ti/Spi}

We here briefly review the flat model of Ti and Spi introduced in \cite{Figueroa-OFarrill:2021sxz} . We particularly emphasise the realisation of a discrete parity symmetry.\\

We will here use the notation $Y^{\mu} \in \mathbb{R}^{n,1}$ to denote vectors of Minkowski space and $\eta_{\mu\nu}$ the corresponding Minkowski metric. Consider the ambient space $\mathbb{R}^{n+1,2}$ and let 
\begin{align}
    \begin{pmatrix}
        Y^+ \\ Y^{\mu} \\ Y^-
    \end{pmatrix} &\in \mathbb{R}^{n+1,2}, &  Y^{\mu} &\in \mathbb{R}^{n,1},
\end{align}
be corresponding vectors. We take the ambient space to be equipped with a metric $h_{IJ}$ of signature $(n+1,2)$ and a preferred vector null $I^{I}$ given by
\begin{align}
    h_{IJ}  &= \begin{pmatrix}
        0 & 0 & 1 \\ 0 &\eta_{\mu\nu} & 0 \\ 1 & 0 & 0
    \end{pmatrix}, & I^{I} = \begin{pmatrix}
        1 \\ 0 \\ 0
    \end{pmatrix}.
\end{align}

The action of the Poincaré group $ISO(n,1)$ on the ambient space is then obtained by considering the subgroup of $SL(n+3)$ preserving $h_{IJ}$ and $I^{I}$. It can be parametrised as
\begin{align}
    M^I {}_J &= \begin{pmatrix}
        1 & - m^{\rho}{}_{\nu} T_{\rho} & -\frac{1}{2} T^2 \\
        0 & m^{\mu}{}_{\nu} & T^{\mu}\\
        0 & 0 & 1
    \end{pmatrix}, & \left( m^{\mu}{}_{\nu} , T^{\mu}\right) &\in SO(3,1) \ltimes \mathbb{R}^{3,1}.
\end{align}
It is typical, when considering this model, to quotient by the action of
\begin{equation}\label{Appdx: ambient parity in the model}
    P \left| \begin{array}{ccc}
         \mathbb{R}^{n+1,2} & \to & \mathbb{R}^{n+1,2} \\
          Y^I & \mapsto &  - Y^I 
    \end{array}\right. .
\end{equation}
We shall not do that here; rather, as we shall see, this map acts as natural discrete symmetry of $\Spi$.

\subsection{Minkowski space compactifications}

From the ambient space, we obtain (two copies of) a trivial bundle over Minkowski space $\mathbb{R} \times \mathbb{M}^{n,1}$ by considering the open subset
\begin{equation}
   E^{\pm} := \big\{ Y^I \in \mathbb{R}^{n+1,2} \; \vert \; I^{I}  Y^{J} h_{IJ}= \pm 1 \big\}
\end{equation}
The isomorphism is straightforwardly given by:
\begin{equation}
    \begin{array}{ccc}
     \mathbb{R} \times \mathbb{M}^{n,1}  & \to &  E^{\pm}  \\
      (u , x^{\mu})    & \mapsto & \begin{pmatrix}
           u \\ x^{\mu} \\ \pm 1
      \end{pmatrix}
    \end{array}.
\end{equation}
Penrose's conformal compactification is recovered by taking $Y^I$ to be in the null cone (i.e. by restricting to $u = \mp\frac{1}{2} Y^\mu Y^{\nu} \eta_{\mu\nu})$ while projective compactification is obtained by quotienting by the lines of $I^I$. If we stick to the ambient space, we obtain instead the extended boundaries. See \cite{RSTA20230042,Borthwick:2024wfn} for further discussions and extension of this ambient picture in the curved case.

\subsection{Extended Boundaries}

From the ambient space, the different boundaries of Minkowski space are obtained by considering 
\begin{equation}
   B := \big\{ Y^I \in \mathbb{R}^{n+1,2} \;\vert \; I^{I}  Y^{J} h_{IJ}= 0 \big\}_{/ \sim}
\end{equation}
where $Y^I \sim \lambda Y^i,\quad  \lambda\in \mathbb{R}^+$. The boundary $B \simeq S^{n+1}$ decomposes as
\begin{equation}
    B = \Ti^{+} \sqcup \Scri^{+} \sqcup \Spi \sqcup \Scri^{-} \sqcup \Ti^{-} \sqcup \left \{ I \right\} \sqcup \left \{ -I \right\},
\end{equation}
where
\begin{align}
\Ti^{+} &:= B \cap \big\{ Y^I \in \mathbb{R}^{n+1,2} \;\vert\; Y^IY^J h_{IJ} = -1, \; Y^{\mu} \; \textrm{future oriented} \big\},  \nonumber\\
   \mathscr{I}^{+} &:= B \cap \big\{ Y^I \in \mathbb{R}^{n+1,2} \;\vert\; Y^IY^J h_{IJ} =0, \; Y^{\mu} \; \textrm{non zero, future oriented} \big\}, \nonumber\\
        \Spi &:= B \cap \big\{ Y^I \in \mathbb{R}^{n+1,2} \;\vert\; Y^IY^J h_{IJ} = 1 \big\},  \\
   \mathscr{I}^{-} &:= B \cap \big\{ Y^I \in \mathbb{R}^{n+1,2} \;\vert\; Y^IY^J h_{IJ} =0, \; Y^{\mu} \; \textrm{non zero, past oriented} \big\},  \nonumber\\
   \Ti^{-} &:= B \cap \big\{ Y^I \in \mathbb{R}^{n+1,2} \;\vert\; Y^IY^J h_{IJ} = -1, \; Y^{\mu} \; \textrm{past oriented} \big\}.   \nonumber
\end{align}
The subset $B \subset \mathbb{R}\textrm{P}^{n+2}$ here is a boundary for $E^{\pm} = \mathbb{R} \times \mathbb{M}^{n,1}$ in the sense that
\begin{equation}
    \mathbb{R}\textrm{P}^{n+2} = E^{+} \sqcup B \sqcup E^-.
\end{equation}
Once again, Penrose's conformal compactification is obtained by intersecting the above with the null cone of the ambient space and projective compactification is obtained by quotienting the above by the lines of $I$. We will now briefly describe coordinates on the different components of the boundary and discuss the action of the discrete symmetry \eqref{Appdx: ambient parity in the model}.

\paragraph{Past and future null infinity}

Adapted coordinates are given by
\begin{align}
    &\left|\begin{array}{ccc}
     \mathbb{R} \times S^{n-1}  & \to &  \Scri^{\pm}  \\
      (u , z^i)    & \mapsto & \begin{pmatrix}
           u \\ Q^{\mu}\left(z^i\right) \\ 0
      \end{pmatrix}
    \end{array}\right. & \text{where}\quad Q^{\mu}\left(z^i\right) &= \begin{pmatrix}
        \pm(1+ z^2) \\ 2 z^i \\ 1-z^2
    \end{pmatrix}.
\end{align}
The discrete symmetry, $P: Y^I \mapsto -Y^I$, sends $\mathscr{I}^{+}$ to $\mathscr{I}^{-}$. This induces, in terms of the above coordinates, an antipodal map of the sphere.

\paragraph{Past and future Ti}

Adapted coordinates are given by
\begin{align}
    &\left|\begin{array}{ccc}
     \mathbb{R} \times H^{n}  & \to &  \Ti^{\pm}  \\
      (u , \y )    & \mapsto & \begin{pmatrix}
           u \\ P^{\mu}(\y) \\ 0
      \end{pmatrix}
    \end{array}\right. & \text{where}\quad P^{\mu}(\y) :=  \begin{pmatrix}
        \pm\sqrt{1 + \y^2}\\ \y
    \end{pmatrix}. 
\end{align}
The discrete symmetry, $P :Y^I \mapsto -Y^I$, sends $\Ti^{+}$ to $\Ti^{-}$. This induces, in terms of the above coordinates, a parity map on the hyperboloid.

\paragraph{Spi}

Adapted coordinates are given by
\begin{align}
    &\left|\begin{array}{ccc}
     \mathbb{R} \times dS_{n}  & \to &  \Spi  \\
      (u , \psi, \vartheta^i )    & \mapsto & \begin{pmatrix}
           u \\ P(\psi, \vartheta)^{\mu} \\ 0
      \end{pmatrix} 
    \end{array}\right. & \text{where}\quad P(\psi, \vartheta)^{\mu} :=  \begin{pmatrix}
        \sinh(\psi) \\[0.4em] \cosh(\psi) \,\vartheta^i
    \end{pmatrix},\quad \vartheta^2=1.
\end{align}
The discrete symmetry, $P: Y^I \mapsto -Y^I$, is here an automorphism of $\Spi$ which induces on $dS_{n-1}$ the usual parity map. In particular, the whole Carrollian geometry must be invariant under this transformation. This is this symmetry that the matching condition at $\Spi$ generalises.

\section{Stationary phase approximation} \label{Section::AppendixSPA}

There appears to be some discrepancy in the literature regarding the phase factor $e^{in\frac{\pi}{4}}$ appearing in the result of the stationary phase approximation (see \eqref{Eq::FallOffScalar}). For this reason, we briefly provide some elements of justification. Since both parts of the integral can be treated in a similar fashion, we will only consider:
\begin{align}
    \frac{m^{n-1}}{2 (2\pi)^n}\int \frac{d^{n}\bm{k}}{\sqrt{1 + \bm{k}^2}} a(\bm{k}) e^{i m \tilde{P} \scal X}.
\end{align}   
The important point is that the phase in the integrand is of the form:
\begin{align}
    m \tilde{P}\scal X = \rho^{-1}\varphi(\bm{k}), \qquad \varphi(\bm{k})=m(-\sqrt{\bm{k}^2+1}\sqrt{1+\bm{y}^2}\sgn(t) + \bm{k}\cdot\bm{y} ). 
\end{align} 
At fixed $X=\rho^{-1}\bm{y}$ lying \emph{within the light-cone of the origin}, $\varphi$ has one critical point at $\bm{k}=\sgn(t)\mathbf{y}$, at which:
\begin{align}
    \varphi(\y) &= - \sgn(t) m  &\textrm{and}&&
    \left(\textrm{Hess}~\varphi(\bm{y})\right)_{\alpha\beta}&=-m\sgn(t)h_{\alpha\beta}.
\end{align} 
In particular, the Hessian is of rank $n$, \emph{definite negative} when $t>0$ and \emph{definite positive} when $t<0$, whilst its determinant $\Delta(\bm{y})=\textrm{Hess}~\varphi(\bm{y})$ satisfies:
\begin{align}
     \sqrt{\lvert\Delta(\bm{y})\rvert}= \frac{m^\frac{n}{2}}{\sqrt{1+\bm{y}^2}}.
\end{align}

Plugging this into the stationary phase approximation \footnote{For a derivation, see for example~\cite[(S.P.) on p6]{Guillemin_Sternberg_1977}}:
\begin{align}
    \int_{\R^n} g(x) e^{ikf(x)} d^nx \underset{k \to+\infty}{\sim} \sum_{x_0} \left(\frac{2 \pi}{k} \right)^\frac{n}{2} |\Delta(x_0)|^{-\frac{1}{2}} e^{i \sigma(x_0) \frac{\pi}{4}} g(x_0) e^{ikf(x_0)} + \smallo{k^{-\frac{n}{2}}}, 
\end{align}
where $\Delta$ and $\sigma$ are respectively the determinant and the signature\footnote{In this context: number of positive eigenvalues - number of negative values} of the Hessian and where the sum is on all critical points (which must be non-degenerate), we obtain in the limit of timelike infinity $(\rho \to 0^+)$,
\begin{align}
     \frac{m^{n-1}}{2 (2\pi)^n}\int \frac{d^{n}\bm{k}}{\sqrt{1 + \bm{k}^2}} a(\bm{k}) e^{i m \tilde{P} \scal X} \sim \frac{m^{\frac{n}{2}-1}}{2(2\pi)^{\frac{n}{2}}}e^{-i \sgn(t) \left( m\rho^{-1} + n\frac{\pi}{4} \right)} a(\bm{y}).
\end{align}

In order to provide some intuition for the phase factor, it can be interesting to review the 1D case, to which multidimensional problems reduce through a diagonalisation argument exploiting the Morse lemma~\cite[Lemma 3.2.3]{Hormander_1971}. For simplicity, suppose that $f \in C_0^{\infty}(\R)$ and consider:
\begin{align}
    \int_\R e^{i\lambda x^2}f(x)\textrm{d} x.
\end{align} 
Since $f$ has compact support the integral is in fact over $[-\delta,\delta]$ for some $\delta >0$. Using Taylor's theorem, we can write $f(x)=f(0)+xf_1(x)$ where $f_1$ is smooth function and:
\begin{align}
    \int_\R e^{i\lambda x}f(x)\textrm{d}x &= \int_{|x|\leq \delta} f(0)e^{i\lambda x^2} \textrm{d}x  + \int_{|x|\leq \delta} f_1xe^{i\lambda x^2}\textrm{d}x.
\end{align}

Integration by parts in the second integral leads to:
\begin{align}
    \int_{|x|\leq \delta} f_1xe^{i\lambda x^2}\textrm{d}x=\bigO{\frac{1}{\lambda}}.
\end{align}

\begin{wrapfigure}{l}{0.25\textwidth}
    \begin{center}
        \begin{tikzpicture}[decoration={markings, mark=at position .5 with {\arrow{stealth}}},scale=2]
            \draw[] (-0.25,0) -- (1.25,0) node[anchor=north]{$\R$};
            \draw[] (0,-0.25) -- (0,0.85) node[anchor=west]{$i \R$};
            \draw[postaction={decorate},thick,color=red] (45:1)--(0,0) node[pos=0.5,anchor=south east]{$C_{\frac{\pi}{4}}$};
            \draw[thick,postaction={decorate},color=red] (1,0) arc (0:45:1) node[pos=0.5,anchor=south west]{$C_{\infty}$};
            \draw[postaction={decorate},thick,color=red] (0,0) -- (1,0) node[pos=0.5,anchor=north]{$C_0$};
        \end{tikzpicture}
    \end{center}
\end{wrapfigure} 

The first integral, on the other hand, can be evaluated using Cauchy's integral theorem for the holomorphic function $z \mapsto e^{i z^2}$ over the contour represented on the left. Indeed:
\begin{align}
\lim_{\lambda \to +\infty} \int_{|x|\leq \delta} f(0)e^{i\lambda x^2} \textrm{d}x &=  \lim_{\lambda \to +\infty} \frac{2 f(0)}{\sqrt{\lambda}}\int_{0}^{\sqrt{\lambda}\delta} e^{iu^2}\textrm{d}u.
    \end{align}

The integral on the $C_\infty$ curve is $\bigO{\frac{1}{\sqrt{\lambda}}}$, and thus:
\begin{align}
    \int_{0}^{\sqrt{\lambda}\delta} e^{iu^2}\textrm{d}u &= e^{i\frac{\pi}{4}}\int_0^{\sqrt{\lambda}\delta}e^{-x^2}\textrm{d}x +\bigO{\frac{1}{\sqrt{\lambda}}}\\&=\frac{\sqrt{\pi}}{2}e^{i\frac{\pi}{4}}+\bigO{\frac{1}{\sqrt{\lambda}}}.
\end{align}
Overall, we obtain the result: 
\begin{align}
    \int_\R e^{i\lambda x^2}f(x)\textrm{d} x \overset{\lambda \to + \infty}{=} \sqrt{\frac{\pi}{\lambda}}f(0)e^{i\frac{\pi}{4}} + \bigO{\frac{1}{\lambda}}.
\end{align}

If a minus sign was present in the exponential, the appropriate contour would pass instead through the lower half-plane, resulting in a phase of $- \frac{\pi}{4}$. This explains why the signature of the Hessian is involved.

\section{Connection and curvature forms for the Levi-Civita connection of the physical metric $\hat{g}$}
In this appendix, for reference and the convenience of the reader, we provide some intermediate computations. We recall that $\hat{g}=\frac{1}{\rho^2}g$ where $g$ admits the asymptotic expansion in equation~\eqref{Ashtekar-Romano: metric Coordinates expansion}. 

Assuming $\nu_0$ is a constant, the connection and curvature forms of the Levi-Civita connection in the coordinate chart $(x^0, x^\alpha)=(\rho, y^\alpha)$ are given by:
 \subsection*{Connection forms: $\omega^i_j := \Gamma^i_{jk}dx^k$}
 \begin{align}
 \omega^0_0&=\left(-\frac{2}{\rho} + \sigma(1-\rho\sigma) + o(\rho)\right)d \rho + \rho(\partial_\alpha \sigma + o(1))dy^\alpha, \\ 
 \omega^0_\alpha & = \rho(\partial_\alpha \sigma +o(1))d\rho  +\rho\nu_0(h_{\alpha \beta}- \rho(3\sigma h_{\alpha \beta} - \frac{1}{2}k_{\alpha \beta})+o(\rho))dy^\beta,\\
 \omega^\alpha_0 & =-\frac{\nu_0^{-1}}{\rho}(\partial^\alpha \sigma - \rho(k^{\alpha\beta}\partial_\beta\sigma-3\sigma\partial^\alpha\sigma)+o(\rho))d\rho+\frac{1}{\rho}\left(-\delta^\alpha_\beta +\rho(\frac{1}{2}k^\alpha_\beta - \sigma \delta^\alpha_\beta)+o(\rho) \right)dy^\beta,\\
 \omega^\alpha_\beta&=\frac{1}{\rho}\left(-\delta^\alpha_\beta +\rho(\frac{1}{2}k^\alpha_\beta-\sigma\delta^\alpha_\beta) +o(\rho)\right)d\rho \\&\hspace{2.3cm} +\left({\Gamma^{(h)}}^\alpha_{\beta\gamma}+\rho\left(D_{(\gamma}k^{\alpha}_{\beta)}-2D_{(\gamma}\sigma\delta^\alpha_{\beta)}-\frac{1}{2}D^\alpha(k_{\beta\gamma}-2\sigma h_{\beta\gamma})\right)+o(\rho)\right)dy^\gamma.
 \end{align}
 \subsection*{Curvature forms: $\Omega^i_j := d\omega^i_j + \omega^i_k\wedge \omega^k_j=\frac{1}{2}R^i_{\,\,jmn}\textnormal{d}x^m\wedge \textnormal{d}x^n$ }
\begin{align}
\Omega^0_{\,\,0} &= o_\alpha(1)d\rho\wedge dy^\alpha + o_{\alpha \beta}(\rho) dy^\alpha\wedge d y^\beta,\\
\Omega^0_{\,\,\alpha}&=\rho\left(-D_\beta D_\alpha \sigma -\sigma\nu_0h_{\beta\alpha} +o(1) \right)d\rho \wedge dy^\beta\\&\hspace{4cm}+\rho^2\nu_0(-D_{(\gamma}k_{\alpha)\delta}+\frac{1}{2}D_\delta(k_{\alpha \gamma} -2\sigma h_{\alpha\gamma})+o(1))dy^\gamma \wedge dy ^\delta, 
\\
\Omega^\alpha_{\,\, 0} &= \frac{\nu_0^{-1}}{\rho}(D_\beta D^\alpha \sigma + \nu_0\sigma \delta^\alpha_\beta + o(1))d\rho \wedge dy^\beta+(\frac{1}{2}D_\gamma k^\alpha_\beta +o(1))dy^\gamma\wedge dy^\beta,
\\
\Omega^\alpha_{\,\,\beta}&=\frac{1}{2}(D_\beta k^\alpha_\gamma-D^\alpha k_{\beta \gamma} +o(1))d\rho\wedge dy^\gamma +(\frac{1}{2}{R^{(h)}}^{\alpha}_{\phantom{\alpha}\beta\gamma\delta}-\nu_0\delta^\alpha_\gamma h_{\beta\delta} )dy^\gamma\wedge dy^\delta\\&\hspace{1.2cm}+\rho\left(-D_{\delta}\left((D_{(\gamma}k^\alpha_{\beta)} -2D_{(\gamma}\sigma\delta^\alpha_{\beta)}-\frac{1}{2}D^\alpha(k_{\beta\gamma}-2\sigma h_{\beta\gamma})) \right) +o(1)\right)dy^\gamma\wedge dy^\delta
\\&\hspace{4.2cm}+\rho\nu_0 \left(\frac{1}{2}(k^\alpha_\gamma h_{\beta\delta} -k_{\beta\delta}\delta^\alpha_\gamma)+2\sigma h_{\beta\delta}\delta^{\alpha}_{\gamma} +o(1)\right)dy^\gamma\wedge dy^\delta.
\end{align}
\subsection*{Ricci tensor: $R^c_{\phantom{c}acb}:= \Omega^i_{\,\,j}(e_i,e_k)\nabla_ ax^j \nabla_bx^k$}
From the above, one can deduce immediately the expression of the Ricci tensor.
\begin{align} R_{00}&=-\frac{\nu_0^{-1}}{\rho}(D^2\sigma +n\nu_0\sigma +o(1)), \\
R_{0\alpha}&=\frac{1}{2}(D_\beta k^\beta_\alpha-D_\alpha k +o(1)), \\
R_{\alpha \beta}&= (R^{(h)}_{\alpha \beta}-(n-1)\nu_0h_{\alpha \beta}) + \rho(\frac{1}{2}D_\alpha D_\beta k +(n-3)D_\alpha D_\beta\sigma - \frac{1}{2}D^2k_{\alpha\beta}\\&\hspace{2cm}+\nu_0(\frac{1}{2}k h_{\alpha\beta}+\frac{1}{2}(n-2)k_{\alpha \beta} +(n-3)h_{\alpha\beta}\sigma ) - {R^{(h)}}_{\gamma(\alpha\phantom{\delta}\beta)}^{\phantom{\gamma(\alpha}\delta}k^\gamma_{\,\delta}) +o(\rho),
\end{align}

\newpage
\bibliographystyle{utPhys}
\bibliography{BibPC3}

\end{document}